%% file: main.tex
\documentclass[letterpaper,11pt]{article}
\usepackage{authblk}
\usepackage{amsmath}
\usepackage{bm}
\usepackage{amsthm}
\usepackage{amssymb}
\usepackage{etoolbox}
\usepackage{mathtools}
\usepackage{algorithm}
\usepackage{algpseudocode}

\usepackage{tikz}
\usetikzlibrary{shapes.geometric,arrows,decorations.markings}
\usepackage{xcolor}
\usepackage{xfrac}
\usepackage{colortbl}
\usepackage{array,arydshln,multirow}
\usepackage{xspace}

\usepackage{paralist}
\usepackage{graphicx}
\usepackage[left=1.0in,top=1.0in,right=1.0in,bottom=1.0in,nohead]{geometry}
\usepackage{wrapfig}
\usepackage{cite}
\usepackage{soul}
\usepackage{thmtools}
\usepackage{thm-restate}

\theoremstyle{definition}

\newtheorem*{theorem*}{Theorem}
\newtheorem{theorem}{Theorem}
\newtheorem{corollary}[theorem]{Corollary}
\newtheorem{lemma}[theorem]{Lemma}

\newtheorem*{prop*}{Proposition}
\newtheorem*{claim*}{Claim}
\newtheorem*{conjecture*}{Conjecture}

\patchcmd{\subequations}{}%
{}{}{}

\graphicspath{{figs/},{new_figs/}}

\makeatletter
\newcommand{\leqnomode}{\tagsleft@true}
\newcommand{\reqnomode}{\tagsleft@false}
\makeatother

\newcommand{\OP}[1]{\operatorname{#1}}
\newcommand{\MC}[1]{{\mathcal #1}}

\newcommand{\nonforbidden}{unsaturated\xspace}
\newcommand{\NonForbidden}{Unsaturated\xspace}
\newcommand{\NFPrszero}{\OP{NExtmPrs}} 

\long\def\longdelete#1{}

\title{\bf Handling LP-Rounding for Hierarchical Clustering and Fitting Distances by Ultrametrics\footnote{
This work was supported in part by the National Science and Technology Council (NSTC), Taiwan, under grants 112-2628-E-A49-017-MY3, 113-2628-E-A49-017-MY3, and 113-2634-F-A49-001-MBK.
This work was partly supported by Institute of Information \& communications Technology Planning \& Evaluation (IITP) grant funded by the Korea government (MSIT) (No. RS-2021-II212068, Artificial Intelligence Innovation Hub).
This work was partly supported by an IITP grant funded by the Korean Government (MSIT) (No. RS-2020-II201361, Artificial Intelligence Graduate School Program (Yonsei University)).
This work was supported by the National Research Foundation of Korea(NRF) grant funded by the Korea government(MSIT) (RS-2025-00563707).
}
}

\author[$\dagger$]{Hyung-Chan An}
\author[$\ddagger$]{Mong-Jen Kao}
\author[$\dagger$]{Changyeol Lee}
\author[$\ddagger$]{Mu-Ting Lee}
\affil[$\dagger$]{Department of Computer Science, Yonsei University, South Korea}
\affil[$\ddagger$]{Department of Computer Science, National Yang-Ming Chiao-Tung University, Taiwan}

\date{}

\begin{document}

\maketitle
\vspace{-2\baselineskip}
\thispagestyle{empty}

\input{0-abstract}

\newpage

\thispagestyle{empty}

\tableofcontents

\newpage

\setcounter{page}{1}

\addtocontents{toc}{\protect\setcounter{tocdepth}{1}}

\input{1-intro}

\addtocontents{toc}{\protect\setcounter{tocdepth}{2}}

\input{2-problem-formulation}

\input{3-LP-rounding}

\input{4-analysis-overview}

\input{5-bound-1}

\input{5-bound-2}

\input{5-bound-3}

\input{5-bound-4}

\input{6-extension}

\input{7-conclusion}

\newpage

\input{8-appendix}

\addcontentsline{toc}{section}{References}

\bibliographystyle{alpha}

\bibliography{lit.bib}
\end{document}

%% file: 0-abstract.tex
\begin{abstract}
We consider the classic correlation clustering problem in the hierarchical setting.
Given a complete graph $G=(V,E)$ and $\ell$ layers of input formation, where the input of each layer consists of a non-negative weight and a labeling of the edges with either $+$ or $-$, this problem seeks to compute for each layer a partition of $V$ such that the partition for any non-top layer subdivides the partition in the upper-layer and the weighted number of disagreements over the layers is minimized, where
the disagreement of a layer is the number of $+$ edges across parts plus the number of $-$ edges within parts.

\smallskip

Hierarchical correlation clustering is a natural formulation of the classic problem of fitting distances by ultrametrics, which is further known as numerical taxonomy~\cite{9140552a-517f-3ea0-a658-de90bd35ae63,SS62,SS63} in the literature.
While single-layer correlation clustering received wide attention since it was introduced in~\cite{DBLP:journals/ml/BansalBC04} and major progress evolved in the past three years~\cite{DBLP:conf/focs/Cohen-AddadLN22,DBLP:conf/focs/Cohen-AddadL0N23,10.1145/3618260.3649749,DBLP:conf/stoc/Cohen-AddadLPTY24}, few is known for this problem in the hierarchical setting~\cite{DBLP:journals/siamcomp/AilonC11,DBLP:journals/jacm/CohenAddadDKPT24}.
The lack of understanding and adequate tools is reflected in the large approximation ratio known for this problem, which originates from 2021.

\smallskip

In this work we make both conceptual and technical contributions towards the hierarchical clustering problem.
We present a simple paradigm that greatly facilitates LP-rounding in hierarchical clustering, illustrated with a delicate algorithm providing a significantly improved approximation guarantee of $25.7846$ for the hierarchical correlation clustering problem.

\smallskip

Our techniques reveal surprising new properties and advances the current understanding for the formulation presented and subsequently used in~\cite{DBLP:journals/siamcomp/AilonC11,DBLP:journals/jacm/CohenAddadDKPT24,DBLP:journals/siamcomp/CohenAddadFLM25,DBLP:conf/soda/CharikarG24} for hierarchical clustering over the past two decades.
This provides a unifying interpretation on the core-technical problem in hierarchical clustering as the problem of finding cuts with prescribed properties regarding the average distance of certain cut pairs.

\smallskip

We further illustrate this perspective by showing that a direct application of the paradigm and techniques presented in this work gives a simple alternative to the state-of-the-art result presented in~\cite{DBLP:conf/soda/CharikarG24} for the ultrametric violation distance problem.

\bigskip

\end{abstract}

{\small \textbf{Keywords:} %
hierarchical correlation clustering,
ultrametric embedding,
correlation clustering,}

{\small \phantom{\textbf{Keywords:} }%
linear programming rounding,
approximation algorithms}

%% file: 1-intro.tex
\section{Introduction}

Clustering is among the central problems in unsupervised machine learning and data mining.
For a given data set and information regarding pairwise similarity of the elements, the general objective is to come up with a partition of the elements into groups such that similar elements are clustered into the same group and dissimilar elements belong to different groups.

\smallskip

{\sc Correlation Clustering}, among various formulations introduced towards the aforementioned objective, has been one of the most successful model since its introduction by Bansal, Blum, and Chawla in~\cite{DBLP:journals/ml/BansalBC04}.
Given a complete graph $G=(V,E)$ and a labeling of the edges with either $+$ or $-$, the goal is to partition the vertices so as to minimize the number of disagreements between the partition computed and the input labels, namely, the number of $+$ edges clustered into different parts plus the number of $-$ edges clustered into the same part.
Due to the simplicity and modularity of this formulation, correlation clustering has found vast applications in practice, e.g., finding clustering ensembles~\cite{DBLP:journals/kais/BonchiGU13}, duplicate detection~\cite{DBLP:conf/icde/ArasuRS09}, community mining~\cite{DBLP:conf/nips/ChenSX12}, disambiguation tasks~\cite{DBLP:journals/tkde/KalashnikovCMN08}, automated labeling~\cite{DBLP:conf/wsdm/AgrawalHKMT09,DBLP:conf/www/ChakrabartiKP08}, and many more.

\smallskip

Various algorithms with an $O(1)$-approximation guarantee exist in the literature for the correlation clustering problem, including classic results in the early 2000s~\cite{DBLP:journals/ml/BansalBC04, DBLP:journals/jcss/CharikarGW05, DBLP:journals/jacm/AilonCN08}, the elegant $2.06$-approximation based on LP-rounding~\cite{DBLP:conf/stoc/ChawlaMSY15}, and recent breakthroughs {that} evolved in the past three years using the Sherali-Adams hierarchy~\cite{DBLP:conf/focs/Cohen-AddadLN22,DBLP:conf/focs/Cohen-AddadL0N23} and 
a strong formulation~\cite{10.1145/3618260.3649749,CCLLLNTVYZ25} known as cluster LP. 
Currently, the best approximation ratio is $1.437+\epsilon$, and $(24/23-\epsilon)$-approximation is NP-hard~\cite{10.1145/3618260.3649749} for any $\epsilon > 0$.

\smallskip

Motivated by the large number of applications in practice, efficient approximation algorithms based on combinatorial approaches have been introduced in the literature, including linear time algorithm~\cite{DBLP:journals/jacm/AilonCN08}, dynamic algorithms~\cite{DBLP:conf/focs/BehnezhadDHSS19}, results for distributed models~\cite{DBLP:conf/soda/CaoHS24,DBLP:conf/focs/BehnezhadCMT22, DBLP:conf/icml/Cohen-AddadLMNP21}, streaming models~\cite{DBLP:conf/soda/BehnezhadCMT23, DBLP:conf/nips/MakarychevC23, DBLP:conf/soda/CambusKLPU24, DBLP:conf/innovations/Assadi022, DBLP:conf/icml/Cohen-AddadLMNP21}, and very recent sublinear time algorithms~\cite{DBLP:conf/innovations/Assadi022, DBLP:conf/stoc/Cohen-AddadLPTY24,CCLLLNTVYZ25}.

\paragraph{Correlation Clustering in the Hierarchical Setting.}

In the hierarchical setting, we are given a complete graph $G=(V,E)$ and $\ell$ layers of input information regarding pairwise similarity of the elements, where the input information for each layer consists of a non-negative weight and a labeling of the edges with either $+$ or $-$.
The goal is to produce for each layer a partition of the elements in $V$ such that (i) the partition for any non-top layer subdivides the partition in the upper layer and (ii) the weighted disagreements over all layers is minimized.

\smallskip

Hierarchical correlation clustering is a natural formulation for the classic problem of fitting given distance information by ultrametrics, which is also known as numerical taxonomy in the literature~\cite{9140552a-517f-3ea0-a658-de90bd35ae63,SS62,SS63,DBLP:conf/approx/HarbKM05,DBLP:journals/siamcomp/AilonC11,DBLP:journals/jacm/CohenAddadDKPT24}.
While single-layer correlation clustering was extensively studied with various types of techniques {in} the past two decades, the multi-layer setting remains much less understood to date. 
The main challenge of this problem has been in the need to produce a sequence of consistent partitioning of the elements subject to the unrelated, possibly conflicting, similarity information given for the layers.

\smallskip

Ailon and Charikar~\cite{DBLP:journals/siamcomp/AilonC11} presented both combinatorial-based and LP-rounding algorithms to obtain a $\min\{ \ell+2, \hspace{1pt} O(\log n \log \log n)\}$-approximation, utilizing the pivot-based algorithm~\cite{DBLP:journals/jacm/AilonCN08} and a region growing argument.
In a breakthrough result for this problem, Cohen-Addad, Das, Kipouridis, Parotsidis, and Thorup~\cite{DBLP:journals/jacm/CohenAddadDKPT24} presented an unconventional approach to {obtain} the first constant factor ($>1000$) approximation using the LP presented in~\cite{DBLP:journals/siamcomp/AilonC11} and state-of-the-art algorithms for single-layer correlation clustering. 
This has remained the best approximation ratio known for this problem since 2021.

\paragraph{Fitting Distance by Ultrametrics (Numerical Taxonomy).}

In the numerical taxonomy problem, we are given measured pairwise distances $\MC{D} \colon {\binom{V}{2}} \mapsto {\mathbb R}_{>0}$ for a set of elements and the goal is to produce a tree metric or an ultrametric $T$ that spans $V$ {and} minimizes the $L_p$-norm
$$\Vert T - \MC{D} \Vert_p \; := \; \left( \; \sum_{\{i,j\}\in {\binom{V}{2}}} \; | \; d_T(i,j) - D(i,j) \; |^p \; \right)^{1/p} \hspace{-12pt} ,$$
where $p$ is a prescribed constant with $1 \le p \le \infty$ and $d_T$ is the distance function for $T$. 

\smallskip

Since Cavalli-Sforza and Edwards introduced the numerical taxonomy problem, it has collected an extensive literature~\cite{9140552a-517f-3ea0-a658-de90bd35ae63,Farris72,WSSB77,Day87}.
While this problem was initially introduced in the $L_2$-norm, Farris~\cite{Farris72} suggested using the $L_1$-norm in 1972.
Further, it is known that for any $1\le p \le \infty$, an algorithm that computes an ultrametric can readily be applied for computing a tree metric  losing a factor of at most $3$ in the approximation guarantee~\cite{DBLP:journals/siamcomp/AgarwalaBFPT99,DBLP:journals/jacm/CohenAddadDKPT24}.

\smallskip

For the $L_\infty$-norm, it is known that an optimal ultrametric can be computed in time proportional to the number of input distance pairs~\cite{DBLP:journals/algorithmica/FarachKW95} and can be approximated in subquadratic time~\cite{DBLP:conf/icml/Cohen-AddadVL21,DBLP:conf/icml/Cohen-AddadSL20}.
For the case with general tree metrics, this problem is APX-hard and $O(1)$-approximation is known~\cite{DBLP:journals/siamcomp/AgarwalaBFPT99}.

\smallskip

For constant $p$ with $1\le p < \infty$, the developments have been slower and remains much less understood to date~\cite{DBLP:journals/jco/MaWZ99,DBLP:conf/approx/Dhamdhere04,DBLP:conf/approx/HarbKM05,DBLP:journals/siamcomp/AilonC11,DBLP:journals/jacm/CohenAddadDKPT24}.
Among them, $L_1$-norm in particular has been extensively studied~\cite{DBLP:conf/approx/HarbKM05,DBLP:journals/siamcomp/AilonC11,DBLP:journals/jacm/CohenAddadDKPT24} and a constant-factor approximation was given by~\cite{DBLP:journals/jacm/CohenAddadDKPT24}.
For $1 < p < \infty$, $O(\log n \log \log n)$ remains the best approximation ratio~\cite{DBLP:journals/siamcomp/AilonC11}.

\smallskip

When the goal is to edit the minimum number of pairwise distances so as to fit into an ultrametric, the problem is known as the ultrametric violation distance problem. 
This problem can be interpreted as numerical taxonomy under the $L_0$-norm and has been actively studied in recent years~\cite{DBLP:conf/allerton/GilbertJ17,DBLP:conf/swat/FanGRSB20,DBLP:journals/algorithmica/FanRB22,DBLP:journals/siamcomp/CohenAddadFLM25,DBLP:conf/soda/CharikarG24} for both metric-fitting and ultrametric-fitting.
For the ultrametric version, the best result is a randomized $5$-approximation~\cite{DBLP:conf/soda/CharikarG24}.

\subsection{Our Result}
We present a simple paradigm which greatly facilitates LP-rounding in hierarchical clustering.
Our main result is a delicate algorithm {for the hierarchical correlation clustering problem with a significantly improved approximation ratio compared to the previously known guarantee~\cite{DBLP:journals/jacm/CohenAddadDKPT24}.}

\begin{theorem} \label{thm-25-7846-approx-hierarchical-correlation-clustering}
There is a $25.7846$-approximation algorithm for the hierarchical correlation clustering problem.
\end{theorem}

Our algorithm shares the same standard LP relaxation used in the literature~\cite{DBLP:journals/siamcomp/AilonC11,DBLP:journals/jacm/CohenAddadDKPT24,DBLP:journals/siamcomp/CohenAddadFLM25,DBLP:conf/soda/CharikarG24} for the hierarchical clustering problems. However, we present a new property of this LP relaxation that allows us to pretend as if the objective has no negative items, intuitively speaking. 
Applying this property causes us to lose the multiplicative factor of up to two.

\smallskip

Our rounding algorithm inherits several key features from the two previous works~\cite{DBLP:journals/siamcomp/AilonC11,DBLP:journals/jacm/CohenAddadDKPT24} with distinguishable technical characteristics,
which we describe in detail in the next section.
Our paradigm further reveal the core-technical problem in hierarchical clustering as the problem of finding cuts with prescribed properties regarding the average distance of a certain subset of cut pairs.
To illustrate this perspective, we show that a direct application of the paradigm and techniques presented in this work leads to an alternative algorithm for the ultrametric violation distance problem that is quite simple to describe and analyze, whose performance guarantee matches the best known~\cite{DBLP:conf/soda/CharikarG24}.

\begin{corollary} \label{cor-5-approx-L0}
There is a deterministic 5-approximation algorithm for the ultrametric violation distance problem.
\end{corollary}

\smallskip

\subsection{Techniques and Discussion}

We begin with a description on the LP formulation and an overview of the approaches introduced in~\cite{DBLP:journals/siamcomp/AilonC11} and~\cite{DBLP:journals/jacm/CohenAddadDKPT24} which handled the rounding problem in very different ways.

\smallskip

The LP-formulation models the clustering decisions via pairwise dissimilarity of the elements which have values within $[0,1]$ and must satisfy the triangle inequality.
Hence, it is instructive to interpret the LP-solutions as distance functions for the elements over the layers.
Furthermore, the distance between any pair of elements satisfies the non-decreasing property top-down over the layers.
Each label given for the element pairs over the layers corresponds to one item in the objective function with a sign being equal to the label itself, i.e., a plus label for an $\{u,v\}$ pair at the $t$-th layer corresponds to an item $x^{(t)}_{\{u,v\}}$ while a minus label gives an item $(1-x^{(t)}_{\{u,v\}})$.
Handling this discrepancy between signs has been the main challenge of this problem.

\smallskip

Following the convention in the literature, we will refer pairs labeled with $+$ to as \emph{edge pairs} and the rest as \emph{non-edge pairs}.

\paragraph{The Techniques in~\cite{DBLP:journals/siamcomp/AilonC11} and~\cite{DBLP:journals/jacm/CohenAddadDKPT24}.}
In~\cite{DBLP:journals/siamcomp/AilonC11}, the hierarchical clustering is obtained in a top-down manner.
This means that the decisions for the algorithm to make in each iteration is how the partition coming from the previous layer above should be subdivided, and the main challenge is to upper-bound the number of disagreements the current clustering decision will cause in all the successive layers below.

\smallskip

To deal with this issue, the authors in~\cite{DBLP:journals/siamcomp/AilonC11} distributed the overall LP value to each element and showed that, whenever a set $P$ in the partition contains a non-edge pair $\{u,v\}$ with a distance at least $2/3$, there always exists an $r \in [0,1/3]$ such that a ball $B$ with radius $r$ to be centered at either $u$ or $v$ will give a cut $C$, such that the weighted disagreements caused by $C$ in all the successive layers below can be upper-bounded by $O(\log \log n) \cdot \log( \OP{Vol}(P) / \OP{Vol}(B) ) \cdot \OP{Vol}(B)$, where $\OP{Vol}(A)$ for any $A \subseteq V$ accounts for the LP value of the edge pairs contained within $A$ over all the successive layers plus the LP value of the elements within $A$.
The proof towards the existence of such a cut utilizes the famous region growing argument presented in~\cite{DBLP:journals/siamcomp/GargVY96} for the multicut problem. 
Summing up the cost over all such cuts gives a guarantee of $O(\log n \log \log n)$.

\smallskip

The approach presented in~\cite{DBLP:journals/jacm/CohenAddadDKPT24} starts from a reduction to the {\sc Hierarchical Cluster Agreement} problem, in which the input for each layer is a pre-clustering of the elements into groups and the goal is to minimize the weighted symmetric difference with the input pre-clustering over the layers.
The authors showed that an algorithm with an $\alpha$-guarantee for the single-layer correlation clustering can readily be applied to obtain a pre-clustering for each layer with a multiplicative loss of $O(\alpha)$ in the overall guarantee.

\smallskip

The obtained instance for the hierarchical cluster agreement problem can be seen as an instance for the hierarchical correlation clustering problem where the intra-pre-cluster pairs act as edge pairs and the inter-pre-cluster pairs act as non-edge pairs.
To handle the LP-solution for this new instance, a procedure called \emph{LP-cleaning} is presented to further {subdivide} the input pre-clusters according to the LP-solution.
This procedure uses a {clever} filtering setting to classify the elements such that, for each pre-cluster, either all the elements are made singleton pre-clusters or only a very small proportion of elements is made so. 
The setting guarantees that the number of ``edge-pairs'' separated in the new pre-clusters can be upper-bounded by the LP-value the fractional solution already has. 
Furthermore, the diameter of the new pre-clusters is guaranteed to smaller than $1/5$.

\smallskip

To obtain the hierarchical clustering, the authors {present} a brilliant approach that handles set-merging in a bottom-up manner, where the set-merging decisions are guided by the non-singleton pre-clusters computed in the above step
and the structure of existing clusters coming from the previous layer.
Roughly speaking, during the process, the algorithm records for each cluster a core subset which comes from a pre-cluster that has a small diameter and contains the majority of elements within the cluster.
To handle the set-merging decisions for a partition coming from the previous layer, the algorithm unconditionally merges for each non-singleton pre-cluster all the clusters whose core subsets have a nonempty intersection with the pre-cluster.
Then the union of the intersections becomes the core subset of the merged set. Using the properties obtained from the LP-cleaning procedure and the set-merging operation, the authors proved a set of cardinality bounds regarding the size of a cluster and its core subset via an involved induction argument.

\paragraph{Our Techniques.}

In this work we present a new paradigm that handles the LP-rounding problem for hierarchical clustering directly.
Our algorithm inherits several key features from the two previous works~\cite{DBLP:journals/siamcomp/AilonC11,DBLP:journals/jacm/CohenAddadDKPT24} with distinguishable technical characteristics.

\smallskip

Our algorithm uses the same LP relaxation used in previous works~\cite{DBLP:journals/siamcomp/AilonC11,DBLP:journals/jacm/CohenAddadDKPT24,DBLP:journals/siamcomp/CohenAddadFLM25,DBLP:conf/soda/CharikarG24}.
Our new, crucial observation is: in any optimal LP solution, the (weighted) number of non-edge pairs with distance strictly smaller than one over the layers is always upper-bounded by the objective value of the LP solution itself.
Hence, whenever the LP-solution pays a nonzero cost to separate a non-edge pair, the cost later incurred by that pair, if any, can readily be attributed to the cost of this LP-solution.

\smallskip

This suggests that we need to handle non-edge pairs with distance one separately since the LP pays nothing for these pairs. We will call them \emph{forbidden pairs.} Our analysis can be intuitively (but not formally) understood as defining a new instance of the problem where the forbidden pairs become non-edge pairs and non-forbidden pairs become edge pairs and then measuring solution costs there. 
This general property avoids the tricky problems in handling the discrepancy between the items with two different signs in the original objective function, greatly facilitating the task of LP-rounding and the analysis in the context of hierarchical clustering.

\smallskip

Our rounding algorithm consists of two components:
(i) A pre-clustering algorithm which takes as input a distance function and produces a partition of the elements which guarantees bounds on both the diameters of the pre-clusters and the average distances of the non-forbidden cut pairs.
(ii) A delicate hierarchical clustering algorithm that handles the set-merging decisions in a bottom-up manner based on the information given by the pre-clusters and the structures of the existing partition coming from the previous layer.

\smallskip

For the first component, our pre-clustering algorithm is a pivot-based algorithm~\cite{DBLP:journals/siamcomp/AilonC11,DBLP:journals/jacm/AilonCN08} that takes an entirely different approach from the pre-clustering algorithm presented in~\cite{DBLP:journals/jacm/CohenAddadDKPT24}
to some extent.
On the other hand, our algorithm starts with a big cluster containing all the elements and iteratively subdividing the clusters until every cluster has a diameter strictly smaller than $1/3$.
When this property is not yet met, an element with an eccentricity at least $1/3$ is picked, and the algorithm either makes the element a singleton cluster or it makes a cut with a ball of radius $1/3-\epsilon$ centered at that element.
This guarantees an average distance at least $1/6$ for the non-forbidden cut pairs.
Hence, the number of non-forbidden cut pairs is bounded by a small factor to the objective value of the LP-solution.
Moreover, we establish this bound using only cut pairs that are not too-far-apart.

\smallskip

Our hierarchical clustering algorithm inherits the guided set-merging framework in~\cite{DBLP:journals/jacm/CohenAddadDKPT24}.
Our algorithm imposes a set-merging condition that captures the elements necessary to provide a good structure for hierarchical clustering yet sufficient to yield a small constant loss in the approximation guarantee. 
We show that, this set-merging condition, combined with the diameter bound for the pre-clusters, leads to a geometrically-decreasing territory of the \emph{non-core part} for any cluster in the hierarchy.
This is the key to a set of substantially stronger cardinality bounds which scales with the core-parameter used in the set-merging condition.

\smallskip

To illustrate another use of our paradigm, we show that a direct application of our pre-clustering algorithm in a top-down manner with a radius parameter of $1/2$ yields a $5$-approximation for the ultrametric violation distance problem.
This provides a simple alternative algorithm to~\cite{DBLP:conf/soda/CharikarG24}, which is obtained via a pivot-based randomized rounding approach top-down recursively.

\smallskip

Our paradigm reveals the nature of the hierarchical clustering problem as a problem of finding cuts with prescribed properties regarding the average distance for a certain subset of its cut pairs. 
The above two results further suggest  that improved approximation results would be possible if stronger properties on the cuts to be computed can be built.
We believe our techniques would easily extend to other variations of hierarchical clustering problems with different objectives.

%% file: 2-problem-formulation.tex
\section{Problem Formulation}

\label{sec-preliminaries}

In the hierarchical correlation clustering problem, we are given a complete graph $G=(V,E)$ and $\ell$ layers of input information $(\delta_1, E^{(1)}), \ldots, (\delta_\ell, E^{(\ell)})$, where $\delta_t \in \mathbb R_{\ge 0}$ is a non-negative weight and $E^{(t)} \subseteq E$ is the set of edges labeled with $+$ at the $t$-th layer.
We refer $E^{(t)}$ and $NE^{(t)} := E \setminus E^{(t)}$ to as the \emph{edge pairs} and the \emph{non-edge pairs} at the $t$-th layer, respectively.
We refer the $1$-st layer to as the \emph{bottom layer} and the $\ell$-th layer to as the \emph{top layer}.

\smallskip

A feasible solution to this problem is a tuple $(\MC{P}^{(1)}, \ldots, \MC{P}^{(\ell)})$, where $\MC{P}^{(t)}$ is a partition of $V$ into groups such that $\MC{P}^{(t)}$ is a subdivision of $\MC{P}^{(t+1)}$ for any $t$ with $1\le t < \ell$. 
That is, for any $P \in \MC{P}^{(t)}$, there always exists $P' \in \MC{P}^{(t+1)}$ such that $P \subseteq P'$.
We say that a collection of partitions $\{\MC{P}^{(t)}\}_{1\le t \le \ell}$ is \emph{consistent} if it satisfies the above property.

\begin{figure*}[bp]
\centering
\fbox{
\begin{minipage}{.86\textwidth}
\vspace{-4pt}
\begin{align}
\text{min} \; & \;\; \sum_{1\le t \le \ell} \delta_t \cdot \left( \; \sum_{\{u,v\} \in E^{(t)}} x^{(t)}_{\{u,v\}} \; + \; \sum_{\{u,v\} \in NE^{(t)}} \left( \; 1-x^{(t)}_{\{u,v\}} \; \right) \; \right) & & \label{LP-HIER-CORR-CLUS} \tag*{LP-(*)} \\[10pt]
\text{s.t.} \; & \;\;\; x^{(t)}_{\{u,v\}} \; \le \; x^{(t)}_{\{u,p\}} \; + \; x^{(t)}_{\{p,v\}},  & & \hspace{-4cm} \forall \; 1\le t\le \ell, \; u,v,p \in V, \notag \\[6pt]
& \;\;\; 0 \; \le \; x^{(t+1)}_{\{u,v\}} \; \le \; x^{(t)}_{\{u,v\}} \; \le \; 1,  & & \hspace{-4cm} \forall \; 1\le t < \ell, \; u,v \in V.  \notag
\end{align}
\vspace{-10pt}
\end{minipage}\quad
}
\caption{An LP formulation for the Hierarchical Correlation Clustering.}
\label{fig-hier-corr-clus-natural-lp}
\end{figure*}

\smallskip

The number of disagreements caused by $\MC{P}^{(t)}$ is defined to be the number of edge pairs in $E^{(t)}$ that result in separated in $\MC{P}^{(t)}$ plus the number of non-edge pairs in $NE^{(t)}$ that are clustered into the same group in $\MC{P}^{(t)}$.
Formally, we use 
$$\#(\MC{P}^{(t)}) \; := \; \sum_{P \in \MC{P}^{(t)}} \left| \left\{ \hspace{1pt} \{p,q\} \in NE^{(t)} \; \colon \; p,q \in P \; \right\} \right| \; + \hspace{-1pt} \sum_{ \substack{ P, P' \in \MC{P}^{(t)} \\[2pt] P \neq P' } } \left| \left\{ \hspace{1pt} \{p,q\} \in E^{(t)} \; \colon \; p \in P,q \in P' \; \right\} \right|$$
to denote the number of disagreements caused by $\MC{P}^{(t)}$.
The goal of this problem is to compute a feasible solution $\{ \MC{P}^{(t)} \}_{1\le t\le \ell}$ that minimizes the weighted disagreements $\sum_{1\le t \le \ell} \hspace{1pt} \delta_t \cdot \#(\MC{P}^{(t)})$.

\medskip

We use the LP formulation in Figure~\ref{fig-hier-corr-clus-natural-lp} from~\cite{DBLP:journals/siamcomp/AilonC11,DBLP:journals/jacm/CohenAddadDKPT24} for the hierarchical correlation clustering problem.
In this formulation, for each $1\le t\le \ell$ and $\{u,v\} \in {\binom{V}{2}}$, we use an indicator variable $x^{(t)}_{\{u,v\}} \in \{0,1\}$ to denote the clustering decision for $u$ and $v$ at the $t$-th layer, i.e., $x^{(t)}_{\{u,v\}} = 0$ if and only if $u,v \in Q$ for some $Q \in \MC{P}^{(t)}$ and $x^{(t)}_{\{u,v\}} = 1$ otherwise.

\smallskip

Since the triangle inequality is satisfied, we will interpret $x^{(t)}$ as a distance function defined on the elements at the $t$-th layer.
Moreover, for each $\{u,v\} \in {\binom{V}{2}}$, the hierarchical constraint requires that $x^{(t)}_{\{u,v\}}$ is non-increasing bottom-up over the layers.
In the rest of this paper we will implicitly assume that $x^{(t)}_{\{u,u\}} = 0$ for any $u \in V$.

\paragraph{Notation.}

We use the following {notation.}
For any $S \subseteq V$, we use $\overline{S}$ to denote $V \setminus S$.
Let $x^{(1)}, \ldots, x^{(\ell)}$ be a feasible solution for~\ref{LP-HIER-CORR-CLUS}.
For any $1\le t\le \ell$, $P,Q \subseteq V$, and $r \in \mathbb R_{\ge 0}$, we use 
$$\OP{Ball}^{(t)}_{<r}(P,Q) \; := \; \left\{ \; v \in Q \; \colon \; \min_{u \in P} \; x^{(t)}_{\{v,u\}} < r \; \right\}$$
to denote the set of elements in $Q$ that are at a distance of strictly less than $r$ from some element in $P$ in the $t$-th layer.
We use $$\OP{diam}^{(t)}(Q) \; := \; \max_{u,v \in Q} \; x^{(t)}_{\{u,v\}}$$ to denote the diameter of the set $Q$ with respect to $x^{(t)}$.

\smallskip

When an arbitrary distance function $x$ is referenced, we use $\OP{Ball}^{(x)}_{<r}(P,Q)$ and $\OP{diam}^{(x)}(Q)$ to denote the same concept with respect to the distance function $x$.

%% file: 3-LP-rounding.tex
\section{LP-Rounding Algorithm}

Solve the~\ref{LP-HIER-CORR-CLUS} in Figure~\ref{fig-hier-corr-clus-natural-lp} for an optimal solution $\tilde{x}$.
For any $1\le t\le \ell$, define 
$$\OP{Fbd}^{(t)} \; := \; \left\{ \; \{p,q\} \in NE^{(t)} \; \colon \; \tilde{x}^{(t)}_{\{p,q\}} = 1 \; \right\}$$ 
to be the set of non-edge pairs with distance one at the $t$-th layer.
We refer these pairs to as \emph{forbidden pairs} since the LP solution pays no cost to separate them.

\smallskip

Our rounding algorithm consists of two parts.
The first part is a pre-clustering algorithm that takes as input a distance function $x$ and produces a partition $\MC{Q}$ with the following two properties.
\begin{enumerate}
	\item
		For each $Q \in \MC{Q}$, the diameter of $Q$ with respect to $x$ is strictly smaller than $1/3$. 
		
	\item
		The \emph{not-too-far-apart pairs} separated by $\MC{Q}$ have a large average distance.
		In particular, those with a distance at most $5/6$ already have an average distance at least $1/6$.
\end{enumerate}
We describe this algorithm later in this section.

\medskip

The second part is a hierarchical clustering algorithm that outputs a consistent partitioning $\{\MC{P}^{(t)}\}_{1\le t\le \ell}$, where each set $P$ in $\MC{P}^{(t)}$ is further associated with a gluer set denoted $\Delta^{(t)}(P)$.

\smallskip

\begin{algorithm*}[htp]
\caption{Hierarchical-clustering$(\{\tilde{x}^{(t)}\}_{1\le t\le \ell})$} \label{algo-hier-clustering}
\begin{algorithmic}[1]
\State $\MC{P}^{(0)} \gets \{ \{v\} { \colon {v \in V}} \}$ and $\Delta^{(0)}(P) \gets P$ for all $P \in \MC{P}^{(0)}$.
\For {$t = 1$ to $\ell$} 
	\State $\MC{P}^{(t)} \gets \emptyset$.
	\State $\MC{Q}^{(t)} \gets $ Pre-clustering$(\tilde{x}^{(t)})$. \quad \texttt{// Pre-clustering from Algorithm~\ref{algo-precluster}.}
	\For {all $Q \in \MC{Q}^{(t)}$} 
		\State Let $\OP{Candi}^{(t)}(Q)$ be the set containing all the sets $P \in \MC{P}^{(t-1)}$ such that 
				\begin{quote}
				\hspace{1.2cm} $\Delta^{(t-1)}(P) \cap P \cap Q \neq \emptyset$ \enskip and \enskip $\left| \; \OP{Ball}^{(t)}_{<2/3} \left( \; P \cap Q, \; P \cap \overline{Q} \; \right) \; \right| < \alpha \cdot \left| P \cap Q \right|$.
				\end{quote}		
		\If { $\OP{Candi}^{(t)}(Q) \neq \emptyset$} \qquad \texttt{// Merge all the sets in $\OP{Candi}^{(t)}(Q)$.}
			\State Let $P_Q := \bigcup_{P \in \OP{Candi}^{(t)}(Q)} P$.
			\State Add $P_Q$ to $\MC{P}^{(t)}$ and set $\Delta^{(t)}(P_Q) \gets Q$.
		\EndIf
	\EndFor
	\For {all $P \in \MC{P}^{(t-1)} \setminus \bigcup_{Q \in \MC{Q}^{(t)}} \OP{Candi}^{(t)}(Q)$} \quad \texttt{// Carry the unmerged sets over to $\MC{P}^{(t)}$}
		\State Add $P$ to $\MC{P}^{(t)}$ and set $\Delta^{(t)}(P) \gets \Delta^{(t-1)}(P)$. 
	\EndFor
\EndFor
\State \Return $\{\MC{P}^{(t)}\}_{1\le t\le \ell}$.
\end{algorithmic}
\end{algorithm*}

\begin{figure*}[b]
\centering
\includegraphics[scale=1]{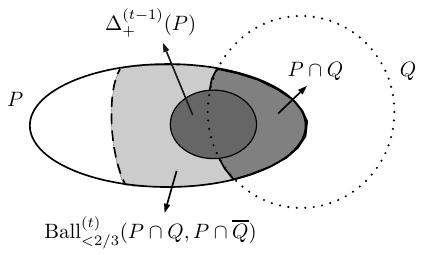}
\qquad
\qquad
\includegraphics[scale=1]{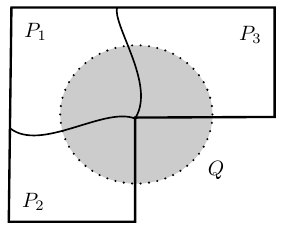}
\caption{
{The setting for intersection requirement~(\ref{ieq-concentrating-requirement}) between $P \in \MC{P}^{(t-1)}$ and $Q \in \MC{Q}^{(t)}$ and the set-merging operation for the sets in $\OP{Candi}^{(t)}(Q)$ with $Q$ being the gluer set.}
}
\label{fig-cluster-merging-condition-1}
\end{figure*}

\smallskip

Let $\MC{P}^{(0)} := \{ \{v \} \}_{v \in V}$ be the initial singleton clustering and define $\Delta^{(0)}(P) := P$ for all $P \in \MC{P}^{(0)}$.
For the $t$-th layer, where $t = 1,2, \ldots,\ell$ in order, the algorithm first applies the pre-clustering algorithm on $\tilde{x}^{(t)}$ to obtain $\MC{Q}^{(t)}$ and iterates over all $Q \in \MC{Q}^{(t)}$.
For each $Q$, the algorithm collects all the sets $P \in \MC{P}^{(t-1)}$ that satisfies the following intersection requirements with $Q$
\begin{equation}
\Delta^{(t-1)}_+(P) \cap Q \neq \emptyset \quad \text{and} \quad \left| \; \OP{Ball}^{(t)}_{<\frac{2}{3}} \left( \; P \cap Q, \; P \cap \overline{Q} \; \right) \; \right| \; < \; \alpha \cdot \left| P \cap Q \right|,
\label{ieq-concentrating-requirement}
\end{equation}
where $\Delta^{(t-1)}_+(P) := \Delta^{(t-1)}(P) \cap P$ will be referred to as the \emph{core} of $P$ and $\alpha := 0.3936$. Note that $\alpha<1/2$.
Refer to Figure~\ref{fig-cluster-merging-condition-1} for an illustration on this condition.

\smallskip

Let $\OP{Candi}^{(t)}(Q)$ denote the sets collected for $Q$.
The algorithm merges all the sets in $\OP{Candi}^{(t)}(Q)$, if it is nonempty, and sets $Q$ to be the gluer set of the merged set.
When all the $Q \in \MC{Q}^{(t)}$ are considered, the algorithm carries all the unmerged sets in $\MC{P}^{(t-1)}$ over to $\MC{P}^{(t)}$ with their gluer sets unchanged.
Refer to Algorithm~\ref{algo-hier-clustering} for a pseudo-code of this algorithm.

\smallskip

Consider the partition $\MC{P}^{(t)}$ computed for any $1\le t \le \ell$.
We refer the sets $\{ P_Q \}_{Q \colon \OP{Candi}^{(t)}(Q) \neq \emptyset}$ to as \emph{newly-created} at the $t$-th layer in the rest of this paper as they are formed as a result of merging the sets in $\OP{Candi}^{(t)}(Q)$ for some $Q \in \MC{Q}^{(t)}$.
On the contrary, the unmerged sets carried over from $\MC{P}^{(t-1)}$ are referred to as \emph{previously-formed}.

\smallskip

Since the distance between any pair is non-increasing bottom-up over the layers, it follows that the diameter of any $Q \in \MC{Q}^{(t')}$ at the $t$-th layer is also strictly smaller than $1/3$ for any $t \ge t'$.
Hence, it follows by construction that
$\OP{diam}^{(t)}(\Delta^{(t)}_+(P)) < \frac{1}{3}$ for any $P \in \MC{P}^{(t)}$ and $1\le t\le \ell$.

\smallskip

The following lemma shows that the candidates to be merged for each $Q \in \MC{Q}^{(t)}$ is unambiguous, and hence Algorithm~\ref{algo-hier-clustering} is well-defined.

\begin{lemma} \label{lemma-hierarchy-friendly}
$\OP{Candi}^{(t)}(Q) \cap \OP{Candi}^{(t)}(Q') = \emptyset$ for any $Q,Q' \in \MC{Q}^{(t)}$ with $Q \neq Q'$.
\end{lemma}

\begin{proof}
Suppose that $P \in \OP{Candi}^{(t)}(Q) \cap \OP{Candi}^{(t)}(Q')$ for some $P \in \MC{P}^{(t-1)}$ and $Q,Q'\in \MC{Q}^{(t)}$.
Let $p \in \Delta^{(t-1)}_+(P) \cap Q$ and $q \in \Delta^{(t-1)}_+(P) \cap Q'$ be two elements.

We have 
$$\tilde{x}^{(t)}_{\{p,q\}} \; \le \; \tilde{x}^{(t-1)}_{\{p,q\}} \; < \; \frac{1}{3}$$ 
by the non-increasing property of the distance function bottom-up over the layers.
Then Condition~(\ref{ieq-concentrating-requirement}) and the diameter bounds of $Q,Q'$ at the $t$-th layer imply that
$$P \cap Q \subseteq \OP{Ball}^{(t)}_{<2/3}(P \cap Q', P \cap \overline{Q'}) \quad \text{and} \quad
P \cap Q' \subseteq \OP{Ball}^{(t)}_{<2/3}(P \cap Q, P \cap \overline{Q}),$$
and hence $|P\cap Q| \; < \; \alpha \cdot |P \cap Q'| \; < \; \alpha^2 \cdot |P \cap Q|$, a contradiction.
\end{proof}

\medskip

\begin{algorithm*}[htp]
\caption{Pre-clustering$(x)$} \label{algo-precluster}
\begin{algorithmic}[1]
	\State Let $\MC{Q} \gets \{V\}$.
	\While {there exists $Q \in \MC{Q}$ with $\OP{diam}^{(x)}(Q) \ge 1/3$}
		\State Pick $(v,Q)$ such that $v \in Q \in \MC{Q}$ and $\max_{u \in Q} x_{\{u,v\}} \ge 1/3$.
		\State $\MC{Q}' \gets$ {\sc One-Third-Refine-Cut}$(Q, v, x)$.
				\label{algo2-cut-pair}
		\State Replace $Q$ with the sets in $\MC{Q}'$ in $\MC{Q}$.
	\EndWhile
\State \Return $\MC{Q}$.
\end{algorithmic}

\smallskip

\hrule

\smallskip

\begin{algorithmic}[1]
\Procedure{One-Third-Refine-Cut}{$Q, v, x$}
	\If {Condition~(\ref{ieq-algo-2-cut-precluster-condition}) is satisfied for $(Q,v)$} 
		\State \Return $\{ \; \{v\}, \; Q\setminus \{v\} \; \}$. \quad \texttt{// make $v$ a singleton}
	\Else
		\State \Return $\{ \; \OP{Ball}^{(x)}_{<1/3}(v,Q), \;\; Q \setminus \OP{Ball}^{(x)}_{< 1/3}(v,Q) \; \}$. \quad \texttt{// cut at $1/3-\epsilon$}
	\EndIf
\EndProcedure
\end{algorithmic}
\end{algorithm*}

Below we describe the pre-clustering algorithm (Algorithm~\ref{algo-precluster}).
The algorithm takes as input a distance function $x$, starts with one big set $\MC{Q} := \{ V \}$, and refines it repeatedly until $\OP{diam}^{(x)}(Q) < 1/3$ for all $Q \in \MC{Q}$.
In each refining iteration, it picks a $Q \in \MC{Q}$ and a vertex $v \in Q$ such that $\max_{u \in Q} x_{\{u,v\}} \ge 1/3$. If
\begin{equation}
\sum_{q \in \OP{Ball}^{(x)}_{<1/3}(v,Q)} x_{\{v,q\}} \;\; \ge \;\; \frac{1}{3} \cdot |\OP{Ball}^{(x)}_{<1/3}(v,Q)| \; - \; \frac{1}{6} \cdot |\OP{Ball}^{(x)}_{<1/2}(v,Q)| - \frac{1}{6},
\label{ieq-algo-2-cut-precluster-condition}
\end{equation}
then the algorithm makes $v$ a singleton pre-cluster by replacing $Q$ with $\{v\}$ and $Q \setminus \{v\}$.
Otherwise, $Q$ is replaced with $\OP{Ball}^{(x)}_{<1/3}(v,Q)$ and $Q \setminus \OP{Ball}^{(x)}_{< 1/3}(v,Q) \vphantom{\text{\Large T}_{\text{\LARGE T}}}$.
We make a note that in~(\ref{ieq-algo-2-cut-precluster-condition}) we use the implicit assumption that $x_{\{u,u\}}=0$ for any $u \in V$ in the distance function $x$.

\medskip

This concludes our rounding algorithm for the hierarchical correlation clustering problem.

%% file: 4-analysis-overview.tex
\section{Overview of the Analysis}

Let $\vphantom{\text{\Large T}_{\text{\large T}}} \{\MC{P}^{(t)}\}_{1\le t\le \ell}$ be the output of Algorithm~\ref{algo-hier-clustering} and $\#(\MC{P}^{(t)})$ be the number of disagreements caused by $\MC{P}^{(t)}$.

\smallskip

Define $\OP{NFbdNE}^{(t)} := NE^{(t)} \setminus \OP{Fbd}^{(t)}$ to be the set of non-forbidden non-edge pairs at the $t$-th layer.
We have that

\begin{equation}
\#(\MC{P}^{(t)}) \; \le \; 
\sum_{P \in \MC{P}^{(t)}} \left( \; \#_{\OP{F}}(P) \; + \; \#_{\OP{NFbdNE}}(P) \; \right)
\; + \;
\sum_{\substack{ P, P' \in \MC{P}^{(t)}, \\[2pt] P \neq P' } } \#_{\OP{NF}}(P,P'),
\vspace{-6pt}
\label{ieq-overall-simp-version-1}
\end{equation}
where 
\begin{itemize}
	\item
		$\#_{\OP{F}}(P) := | \{ \; \{i,j\} \in \OP{Fbd}^{(t)} \colon i,j \in P \; \} |$
		is the number of forbidden pairs clustered within $P$,

	\item
		$\#_{\OP{NFbdNE}}(P) := | \{ \; \{i,j\} \in \OP{NFbdNE}^{(t)} \colon i,j \in P \; \} |$ 
		is the number of non-forbidden non-edge pairs clustered within $P$, and
		
	\item
		$\#_{\OP{NF}}(P,P') := | \{ \; \{i,j\} \notin \OP{Fbd}^{(t)}  \colon  i \in P, j \in P' \; \} |$
		is the number of non-forbidden pairs between $P$ and $P'$.
\end{itemize}

\begin{figure*}[h]
\centering
\includegraphics[scale=0.76]{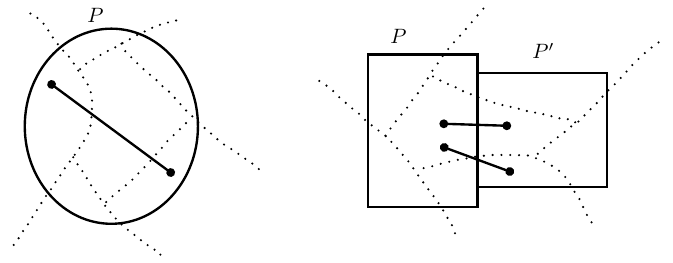}
\caption{Two types of disagreements we will {focus on} -- (a) Forbidden pairs clustered into the same part $P$. (b) Non-forbidden pairs across different parts $P$ and $P'$. 
}
\label{fig-disagreements-all}
\end{figure*}

\noindent
Recall that $\tilde{x}$ is an optimal solution to~\ref{LP-HIER-CORR-CLUS}.
To bound the weighted disagreements, we use a rather surprising property, proved in Section~\ref{section-relating-objectives}.

\begin{lemma}[Section~\ref{section-relating-objectives}] \label{lemma-opt-sols-nonforbidden-pair-value}
$$\sum_{1\le t\le \ell} \; \delta_t \cdot \left( \; \sum_{ \{u,v\} \in {E}^{(t)} } \tilde{x}^{(t)}_{u,v} \; + \; \hspace{-4pt} \sum_{ \{u,v\} \in NE^{(t)} } \left( \; 1 - \tilde{x}^{(t)}_{u,v} \; \right) \; \right) \; \ge \; \sum_{1\le t\le \ell} \; \delta_t \cdot \left| \; \OP{NFbdNE}^{(t)} \; \right|. $$
\end{lemma}

It follows from Lemma~\ref{lemma-opt-sols-nonforbidden-pair-value} that the weighted disagreements caused by the pairs in $\OP{NFbdNE}^{(t)}$, if any, can readily be attributed to the cost of the optimal LP solution.
Furthermore, they can be treated as if they were edge pairs when necessary.

\smallskip

Next we bound $\sum_P \#_{\OP{F}}(P) $ and $\sum_{P \neq P'} \#_{\OP{NF}}(P,P')$ in terms of $|\OP{NFPrs}(\MC{Q}^{(t)})|$, where we use 
$$\OP{NFPrs}(\MC{Q}^{(t)}) \; := \; \left\{ \; \{i,j\} \notin \OP{Fbd}^{(t)} \; \colon \; \{i,j\} \text{ separated in } \MC{Q}^{(t)} \; \right\}$$ to denote the set of non-forbidden pairs that are separated in $\MC{Q}^{(t)}$.
As for $\#_{\OP{F}}(P)$, we prove the following lemma in Section~\ref{sec-counting-disagreements} and~\ref{sec-appendix-forbidden-pairs-within-P}.

\begin{lemma}[Section~\ref{sec-counting-disagreements}, Section~\ref{sec-appendix-forbidden-pairs-within-P}] \label{lemma-number-of-disagreements-forbidden-pairs-within-P}
For $\alpha := 0.3936$ and any $P \in \MC{P}^{(t)}$, we have
$$\#_{\OP{F}}(P) \; \le \; \frac{(2-\alpha)(1+\alpha)^2}{2(1-\alpha)^2} \cdot \beta \cdot \left| \OP{NFPrs}(\MC{Q}^{(t)}, P) \right|,$$
where $\beta := 0.8346$ and $\OP{NFPrs}(\MC{Q}^{(t)}, P) := \left\{ \; \{i,j\} \in \OP{NFPrs}(\MC{Q}^{(t)}) \; \colon \; i,j \in P \; \right\} $
denotes the set of pairs in $\OP{NFPrs}(\MC{Q}^{(t)})$ that reside within $P$.
\end{lemma}

\smallskip

For $\#_{\OP{NF}}(P,P')$, we prove the following lemma.

\begin{lemma}[Section~\ref{sec-counting-disagreements}] \label{lemma-number-of-disagreements-restate-non-forbidden-pairs-P-Pp}
For any $P,P' \in \MC{P}^{(t)}$ with $P \neq P'$, we have
$$\#_{\OP{NF}}(P,P') \; \le \; \max \left\{ \; \frac{1}{1-\alpha}, \; \frac{1+\alpha}{\alpha} \; \right\} \cdot \left| \OP{NFPrs}(\MC{Q}^{(t)}, P, P') \right|,$$
where 
$\OP{NFPrs}(\MC{Q}^{(t)}, P, P') := \left\{ \; \{i,j\} \in \OP{NFPrs}(\MC{Q}^{(t)}) \; \colon \; i \in P, \; j \in P' \; \right\}$ denotes the set of pairs in $\OP{NFPrs}(\MC{Q}^{(t)})$ that are between $P$ and $P'$.
\end{lemma}

Lemma~\ref{lemma-number-of-disagreements-forbidden-pairs-within-P} and Lemma~\ref{lemma-number-of-disagreements-restate-non-forbidden-pairs-P-Pp} bound $\sum_P \#_{\OP{F}}(P) {+} \sum_{P \neq P'} \#_{\OP{NF}}(P,P')$ 
in terms of $|\OP{NFPrs}(\MC{Q}^{(t)})|$.
To further bound $| \OP{NFPrs}(\MC{Q}^{(t)}) |$, we show that the non-forbidden pairs separated in $\MC{Q}^{(t)}$ have an average distance at least $1/6$ via a stronger statement.

\begin{lemma}[Section~\ref{sec-proof-avg-distance-1-3}] \label{lemma-non-forbidden-pairs-separated-avg-cost-1-3}
Consider line~\ref{algo2-cut-pair} in Algorithm~\ref{algo-precluster} with input distance function $x$.
Let $v$ be the pivot chosen in that iteration and $(Q_1, Q_2)$ with $v \in Q_1$ be the pair returned by the procedure {\sc One-Third-Refine-Cut}.
Then
$$
\sum_{ \substack{ \{i,j\} \in \OP{NFPrs}(Q_1, Q_2), \\[2pt] i \in Q_1, \\[2pt] j \in \OP{Ball}^{(x)}_{<1/2}(v, Q_2) } } \hspace{-2pt} \left( \; \min\left\{ \; x_{\{v,j\}}, \; \frac{1}{3} \;\right\} \; - \; x_{\{v,i\}} \; \right) \;\; \ge \;\;
\frac{1}{6} \; \cdot \; \left| \left\{ \; \substack{ \{i,j\} \in \OP{NFPrs}(Q_1, Q_2), \\[2pt] i \in Q_1, \; \\[2pt] j \in \OP{Ball}^{(x)}_{<1/2}(v, Q_2) } \; \right\} \right|,
$$
where $\OP{NFPrs}(Q_1, Q_2)$ denotes the set of non-forbidden pairs between $Q_1$ and $Q_2$.
\end{lemma}

Since $| x_{\{v,i\}} - x_{\{v,j\}}|$ is a lower-bound for $x_{\{i,j\}}$ for any $i,j \in Q_1 \cup Q_2$ by the triangle inequality, Lemma~\ref{lemma-non-forbidden-pairs-separated-avg-cost-1-3} guarantees an average distance at least $1/6$ for the pairs in $\OP{NFPrs}(Q_1, Q_2)$.
Moreover, although the actual distance of such pairs can be much larger than the average, the statement ensures that only a reasonably small amount of it is charged to establish the bound.

\medskip

Using Lemma~\ref{lemma-non-forbidden-pairs-separated-avg-cost-1-3}, we bound 
$\left| \OP{NFPrs}(\MC{Q}^{(t)}) \cap E^{(t)} \right|$ in terms of the objective value of $\tilde{x}^{(t)}$.
Combining all the above with Inequality~(\ref{ieq-overall-simp-version-1}), we obtain the following lemma in Section~\ref{sec-overall}.

\begin{lemma}[Section~\ref{sec-overall}] \label{lemma-overall-approximation-ratio}
\begin{align*}
& {\sum_{1\le t \le \ell} \delta_t \cdot \#(\MC{P}^{(t)})}
\; \le \; \left( \; 7 c(\alpha) +1 \; \right) \cdot \sum_{1\le t\le \ell} \delta_t \cdot \left( \sum_{\{u,v\} \in E^{(t)}} \tilde{x}^{(t)}_{\{u,v\}} + \sum_{\{u,v\} \notin E^{(t)}} \left( 1-\tilde{x}^{(t)}_{\{u,v\}} \right) \right),
\end{align*}
where $c(\alpha) := \max \left\{ \frac{\beta(2-\alpha)(1+\alpha)^2}{2 (1-\alpha)^2}, \frac{1}{1-\alpha}, \frac{1+\alpha}{\alpha} \right\} \approx 3.5406$ for $\alpha := 0.3936$, and $\beta := 0.8346$.
\end{lemma}

\smallskip

This yields the approximation guarantee of {$25.7846$.}

%% file: 5-bound-1.tex
\section{Bounding the Weighted Disagreements}

In this section we provide the proofs for the lemmas described in the previous section.

\subsection{Cardinality Bounds for $P \in \MC{P}^{(t)}$}

\label{sec-cardinality-bounds}

To prove Lemma~\ref{lemma-number-of-disagreements-forbidden-pairs-within-P} and Lemma~\ref{lemma-number-of-disagreements-restate-non-forbidden-pairs-P-Pp}, one of the key ingredients is a set of cardinality bounds regarding the territory of any cluster in terms of its core.

In particular, the intersection requirement in~(\ref{ieq-concentrating-requirement}) leads to a decrease of the non-core territory in a geometric order for any cluster in the hierarchy.

\smallskip

Let $P \in \MC{P}^{(t)}$ be a cluster in the $t$-th layer.
Recall that $\Delta^{(t)}(P)$ denotes the gluer set of $P$ and $\Delta^{(t)}_+(P) := P \cap \Delta^{(t)}(P)$ is referred to as the core set of $P$.
Additionally define 
\begin{itemize}
	\item
		$\OP{Ext}^{(t)}(P) := P \setminus \Delta^{(t)}_+(P)$ to be the extended part of $P$,
		
	\item
		$\ell(t,P)$ to be the top-most layer up to the $t$-th layer at which $P$ is newly-created, and 
		
		{$L^{(t)}_1(P)$} to be the elements in the $2/3$-vicinity of $P' \cap \Delta^{(\ell(t,P))}_+(P)$ within $P'$ at the $\ell(t,P)$-th layer over all $\vphantom{\text{\Large T}_\text{\Large T}} P' \in \OP{Candi}^{(\ell(t,P))}( \Delta^{(\ell(t,P))}(P) )$.
		
		Formally, 
		$$L^{(t)}_1(P) \;\; := \; \bigcup_{ P' \in \OP{Candi}^{(\ell(t,P))}( Q_P ) } \OP{Ball}^{(\ell(t,P))}_{<2/3} \left( P' \cap Q_P, \hspace{1pt} P' \cap \overline{Q_P} \right),$$
		where we use $Q_P := \Delta^{(\ell(t,P))}(P)$ to denote the gluer set of $P$ at the $\ell(t,P)$-th layer. 
        We note that $\ell(t,P)$ is always well-defined.
				
		Refer to the figure below for an illustration.
		Note that it follows that $|L^{(t)}_1(P)| < \alpha \cdot | \Delta^{(t)}_+(P)|$ by the merging condition in Algorithm~\ref{algo-hier-clustering}.

\end{itemize}

\begin{figure*}[h]
\centering
\includegraphics[scale=0.8]{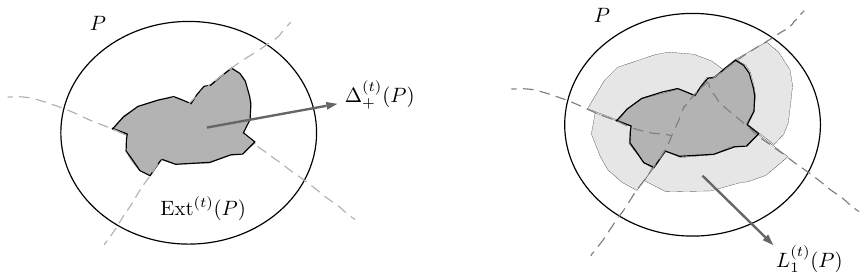}
\end{figure*}

\noindent
We prove the following helper lemma regarding the cardinality of the extended part of $P$ and the reasonably dense structure in any $2/3$-vicinity of it.
The statements are proved using the intersection requirement~(\ref{ieq-concentrating-requirement}) in Algorithm~\ref{algo-hier-clustering} and the diameter bound of $1/3$ for each pre-cluster.

\begin{lemma} \label{lemma-ieq-ext-delta-plus-with-k}
Let $P \in \MC{P}^{(t)}$ be a cluster.
We have
$$|\OP{Ext}^{(t)}(P)| \; \le \; {\min}\left\{ \;\; \frac{\alpha}{1-\alpha} \cdot |\Delta^{(t)}_+(P)|, \;\; \frac{1}{1-\alpha} \cdot |L^{(t)}_1(P)| \;\; \right\}.$$
Furthermore, for any nonempty $A \subseteq \OP{Ext}^{(t)}(P)$, there exists $K^{(t)}_P(A) \subseteq \OP{Ball}^{(t)}_{<2/3}(A, P \setminus A)$ such that 
$$|A| \; \le \; \frac{\alpha}{1-\alpha} \cdot |K^{(t)}_P(A)|.$$
\end{lemma}

\smallskip

We prove the statements in Lemma~\ref{lemma-ieq-ext-delta-plus-with-k} separately.
Note that it suffices to prove the statements for the $\ell(t,P)$-th layer.
Hence, in the following we assume that $P$ is newly-created at the $t$-th layer.

\smallskip

Consider a tree $\MC{T}_P$ built to represent the sequence of set-merging processes leading to $P$, where each node $v \in \MC{T}_P$ is associated with the following two auxiliary information.
\begin{enumerate}
	\item
		$H(v)$ which is a cluster newly-created at the $t'$-th layer for some $t' \le t$.
		Literally this will be the set to which the node $v$ corresponds.
				
	\item
		$\ell(v)$ which is an index of a layer at which $H(v)$ is newly-created. Refer to the construction described below.

\end{enumerate}

\smallskip

We define $\MC{T}_P$ by describing a procedure to construct it. The process starts with a singleton tree with the root node $r$ such that $H(r) := P$ and $\ell(r) := t$.
In each of the iterations that follow, consider the set of current leaf nodes $v$ in $\MC{T}_P$ with $\ell(v) > 1$.
For each of such leaf nodes $v$, consider the sets contained in $\OP{Candi}^{(\ell(v))}(\Delta^{(\ell(v))}(v))$.
For each $P' \in \OP{Candi}^{(\ell(v))}(\Delta^{(\ell(v))}(v))$, create a node for $P'$, say, $u$, as a child node of $v$.
Set $H(u) := P'$ and $\ell(u)$ to be the largest index between $1$ and $\ell(v)$ such that $P'$ is newly-created at the $\ell(u)$-th layer.

\smallskip

\begin{figure*}[h]
\centering
\includegraphics[scale=0.7]{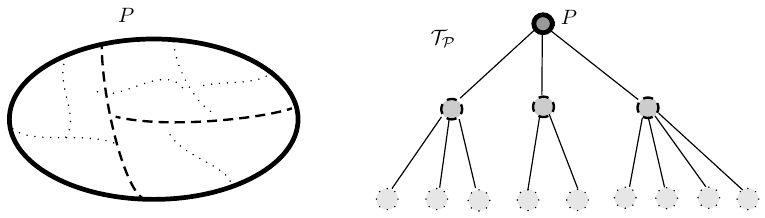}
\end{figure*}

\begin{proof}[{Proof of Lemma~\ref{lemma-ieq-ext-delta-plus-with-k}, Part~I}]
Use a pre-order traversal on $\MC{T}_P$ to define a set of layers as follows.
Initially, define $A_1 := \Delta^{(t)}_+(P)$ and $\OP{Base}_1 := A_1 \cup L_1$, where
$$L_1(P) \;\; := \; \bigcup_{ P' \in \OP{Candi}^{(t)}( \Delta^{(t)}(P) ) } \OP{Ball}^{(t)}_{<2/3} \left( P' \cap \Delta^{(t)}(P), \; P' \cap \overline{\Delta^{(t)}(P)} \right).$$
The traversal starts with the root node $P$ and the initial index $i = 1$.
For any vertex $v$ encountered during the traversal, process $v$ as follows.
If $\Delta^{(\ell(v))}_+(H(v)) \subseteq \OP{Base}_i$, then nothing needs to be done. In this case we proceed to the next vertex directly.

\smallskip

\begin{figure*}[h]
\centering
\includegraphics[scale=0.84]{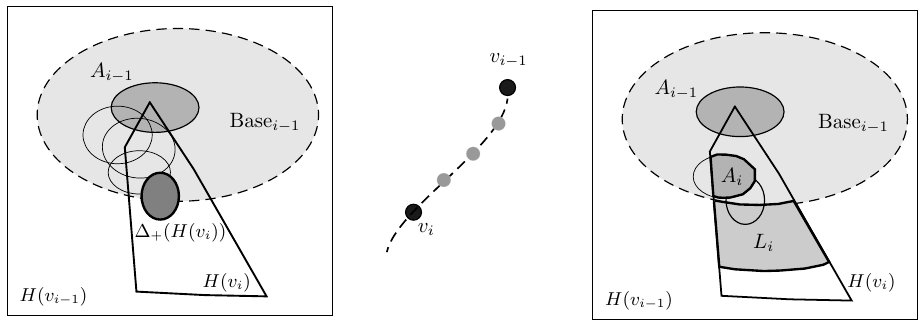}
\caption{From $v_{i-1}$ on which $A_{i-1}$ and $\OP{Base}_{i-1}$ are defined, identify the first descendant $v_i$ whose core set is not contained within $\OP{Base}_{i-1}$. Then $A_i$, $L_i$, and $\OP{Base}_i$ are defined accordingly. 
}
\label{fig-core-tree-S-path1}
\end{figure*}

On the other hand, if $\Delta^{(\ell(v))}_+(H(v)) \not\subseteq \OP{Base}_i$, then consider the parent node $p(v)$ of $v$ in $\MC{T}_P$.
Such node exists since $\Delta^{(\ell(r))}_+(H(r)) \subseteq \OP{Base}_1$ for the root node $r$.
Let $Q(v) := \Delta^{(\ell(p(v))}(H(p(v)))$ denote the gluer set of $H(p(v))$.  
Increase $i$ by one and define
$$A_i \; := \; H(v) \cap Q(v) 
\quad \text{and} \quad
\OP{Base}_i \; := \; \OP{Base}_{i-1} \; \cup \; \OP{Ball}^{(\ell(p(v)))}_{<2/3} \left( \; A_i, \; H(v) \cap \overline{A_i} \; \right). 
$$
Note that $A_i\neq\emptyset.$
Also refer to Figure~\ref{fig-core-tree-S-path1} for an illustration for the definitions.
For any $i \ge 2$, define $$L_i := \OP{Base}_i \setminus \OP{Base}_{i-1}.$$

\smallskip

For any index $i \ge 2$, let $v_i$ denote the specific vertex at which the sets $A_i$, $\OP{Base}_i$, and $L_i$ are defined during the pre-order traversal.
For $i = 1$, define $v_1$ to be the root node $r$ for consistency.
Also refer to Figure~\ref{fig-core-tree-S-path1} for an illustration of the definitions.

\medskip

We prove the following {two invariant conditions} regarding the sets defined during the traversal.
\begin{enumerate}
	\item[i.]
		$|L_i| \le \alpha |A_i|$ for any $i \ge 1$.
		
	\item[ii.]
		$A_i \cap A_j = \emptyset$ for any $i \neq j$. 
		
\end{enumerate}

\smallskip

\noindent
For condition~(i), it suffices to consider any $i \ge 2$.
By the definition of $L_i$ we have 
$$L_i \; \subseteq \; \OP{Ball}^{(\ell(p(v_i)))}_{<2/3}\left( \; A_i, \; H(v_i) \cap \overline{A_i} \; \right) \; = \; \OP{Ball}^{(\ell(p(v_i)))}_{<2/3}\left( \; H(v_i) \cap Q(v_i), \; H(v_i) \cap \overline{Q(v_i)} \; \right),$$
where we recall that $Q(v_i)$ denotes the gluer set of $H(p(v_i))$.
Since $Q(v_i)$ results in the merge of $H(v_i)$, condition~(\ref{ieq-concentrating-requirement}) is satisfied between $H(v_i)$ and $Q(v_i)$.
Hence, $|L_i| \le \alpha |A_i|$.

\smallskip

For condition~(ii), consider any $i,j$ with $1 \le j < i$.
Let $v_k$ be the least common ancestor of $v_i$ and $v_j$ in $\MC{T}_P$.
If $v_k \notin \{ v_i, v_j \}$, then $v_i$ and $v_j$ belong to different subtrees rooted at $v_k$.
Since the sets to which the children nodes of $v_k$ correspond form a partition of $H(v_k)$, it follows that $H(v_i) \cap H(v_j) = \emptyset$ and this condition holds.

\smallskip

Now consider the other case where $v_k \in \{ v_i, v_j\}$, in which $v_j$ is a {proper} ancestor of $v_i$.
Since the core of $v_i$ is not contained within $\OP{Base}_{i-1}$, there exists an element $q \in \Delta^{(\ell(v_i))}_+(H(v_i)) \setminus \OP{Base}_{i-1}$.
Observe that, since $A_i$ intersects $\Delta^{(\ell(v_i))}_+(v_i)$, by the diameter bounds of $Q(v_i)$ and $\Delta^{(\ell(v_i))}_+(v_i)$ together with the triangle inequality, we have
\begin{equation}
    \tilde{x}^{(\ell(p(v_i)))}_{\{q,w\}} \; < \; \frac{2}{3} \quad \text{for any $w \in A_i$}.
    \label{ieq-cardinality-bound-condition-2-1}
\end{equation}
Define 
$$A'_j := \begin{cases}
\; A_j, & \text{if $j > 1$,} \\[2pt]
\; A_1 \cap H(v'), \text{ where $v'$ is the child of $v_1$ such that $A_i \subseteq H(v')$}, & \text{otherwise.}
\end{cases}$$
Note that, to prove that $A_i \cap A_j = \emptyset$, it suffices to prove the statement for $A_i$ and $A'_j$.
For this, we prove the following claim.

\begin{claim*}
$$\tilde{x}^{(\ell(p(v_i)))}_{\{q,u\}} \; \ge \;\; \frac{2}{3} \quad \text{for any $u \in A'_j$}.$$
\end{claim*}

\begin{proof}
Consider the case for which $j \ge 2$.
We have
$$\OP{Base}_j \subseteq \OP{Base}_{i-1} \quad \text{and} \quad H(v_i) \subseteq H(v_j),$$
which shows that $q \in H(v_j) \setminus \OP{Base}_j$.
Since $j\ge 2$, from the construction, we have  that
$$\OP{Ball}^{(\ell(p(v_j)))}_{<2/3} \left( \; A_j, \; H(v_j) \cap \overline{A_j} \; \right) \; \subseteq \; \OP{Base}_j.$$
This implies that
$$\tilde{x}^{(\ell(p(v_i)))}_{\{q,u\}} \; \ge \;\; \tilde{x}^{(\ell(p(v_j)))}_{\{q,u\}} \; \ge \;\; \frac{2}{3} \quad \text{for any $u \in A_j$},$$
where the first inequality follows from the monotonic property of the distances over the layers.
Since $A'_j := A_j$ when $j > 1$, we are done with this case.

\smallskip

For the other case with $j = 1$, recall that $v'$ is the child of $v_1$ such that $A_i \subseteq H(v')$.
From the construction, we have 
$$\OP{Ball}^{(\ell(p(v')))}_{<2/3} \left( \; A_1', \; H(v_1) \cap \overline{A_1'} \; \right) \; \subseteq \; \OP{Base}_1,$$
and hence 
$$\tilde{x}^{(\ell(p(v')))}_{\{q,u\}} \; \ge \;\; \frac{2}{3} \quad \text{for any $u \in A'_1$},$$
Note that since $\ell(p(v'))=\ell(v_1)\ge \ell(p(v_i))$, the monotonicity over the layers completes the proof.
\end{proof}

Combining the above claim with~(\ref{ieq-cardinality-bound-condition-2-1}), we have $$\tilde{x}^{(\ell(p(v_i)))}_{\{u,w\}} \; \ge \;\; \tilde{x}^{(\ell(p(v_i)))}_{\{q,u\}} \; - \; \tilde{x}^{(\ell(p(v_i)))}_{\{q,w\}} \; > \; 0.$$
Since this holds for any $w \in A_i$ and any $u \in A'_j$, we have $A_i \cap A'_j = \emptyset$ and condition~(ii) follows.

\medskip

We are ready to prove the statement of this lemma.
It follows from the above definitions that
$L_i \cap L_j  =  \emptyset$ for all $i \neq j$ and $|\OP{Ext}^{(t)}(P)|  =  \sum_{i \ge 1} |L_i|.$
From invariant condition~(i) and~(ii), we obtain that
\begin{align*}
|\OP{Ext}^{(t)}(P)| \;\; = \;\; \sum_{i \ge 1} |L_i| \;\; 
= \;\; & |L_1| \; + \; \sum_{i \ge 2} |L_i|  \\[2pt]
\le \;\; & |L_1| \; + \; \alpha\cdot\sum_{i \ge 2} |A_i| \;\; \le \;\; |L_1| \; + \; \alpha \cdot |\OP{Ext}^{(t)}(P)|.
\end{align*}
This gives $|\OP{Ext}^{(t)}(P)| \le \frac{1}{1-\alpha} \cdot |L_1|$. 
The first part of the lemma follows from $|L_1| \le \alpha \cdot |\Delta^{(t)}_+(P)|$
\end{proof}

In the following we complete the second part of Lemma~\ref{lemma-ieq-ext-delta-plus-with-k}.

\begin{proof}[{Proof of Lemma~\ref{lemma-ieq-ext-delta-plus-with-k}, Part~II}]
Consider the tree $\MC{T}_P$ and the set of nodes $u \in \MC{T}_P$ whose core set intersects $A$ and whose every ancestor node has its core set being disjoint with $A$.
Formally,
$$\Delta^{(\ell(u))}_+(H(u)) \cap A \neq \emptyset \quad \text{ and } \quad \Delta^{(\ell(v))}_+(H(v)) \cap A = \emptyset \enskip \text{for any ancestor $v$ of $u$ in $\MC{T}_P$}.$$
Let $u_1, u_2, \ldots, u_k$ be the set of all such nodes and $p_1, p_2, \ldots, p_m$ be the parent nodes of $\{u_i\}_{1\le i\le k}$. 
Note that $m \le k$ since some nodes in $\{u_i\}_{1\le i\le k}$ may share a common parent.

\smallskip

It follows that $A \subseteq \bigcup_{1\le i\le m} \OP{Ext}^{(\ell(p_i))}(H(p_i))$ and hence 
$$|A| \; \le \; \frac{\alpha}{1-\alpha} \cdot \sum_{1\le i\le m} \left| \Delta^{(\ell(p_i))}_+(H(p_i)) \right|$$
by the bound proved above for the first part of this lemma.
Furthermore, for each $p_i$, there exists $u_j$ which is a child node of $p_i$ such that 
$$\Delta^{(\ell(p_i))}_+(H(p_i)) \; \cap \; \Delta^{(\ell(u_j))}_+(H(u_j)) \; \neq \; \emptyset \quad \text{and} \quad 
A \cap \Delta^{(\ell(u_j))}_+(H(u_j)) \neq \emptyset.$$
Since $\max\left\{ \; \OP{diam}^{(t)}(\Delta^{(\ell(p_i))}_+(H(p_i))), \; \OP{diam}^{(t)}( \Delta^{(\ell(u_j))}_+(H(u_j)) ) \; \right\} < 1/3$ by the monotonic property of the distance functions over the layers and the diameter bound of the core sets, it follows that
$$\bigcup_{1\le i\le m} \Delta^{(\ell(p_i))}_+(H(p_i)) \;\; \subseteq \;\; \OP{Ball}^{(t)}_{<2/3}(A, P \setminus A).$$
From the construction, $\OP{Ext}^{(\ell(p_i))}(H(p_i))$ for all $1\le i \le m$ are disjoint. Hence, taking $K^{(t)}_P(A) := \bigcup_{1\le i\le m} \Delta^{(\ell(p_i))}_+(H(p_i))$ completes the proof of this lemma.
\end{proof}

\medskip

%% file: 5-bound-2.tex
\subsection{Counting the Number of Disagreements}

\label{sec-counting-disagreements}

We count the total number of disagreements in $\MC{P}^{(t)}$ in terms of the number of pairs in $\OP{NFPrs}(\MC{Q}^{(t)})$ for Lemma~\ref{lemma-number-of-disagreements-forbidden-pairs-within-P} and Lemma~\ref{lemma-number-of-disagreements-restate-non-forbidden-pairs-P-Pp}.
Recall that, for any $P \in \MC{P}^{(t)}$ and any $P' \in \MC{P}^{(t)}$, $P \neq P'$, 
\begin{itemize}
	\item
		$\#_{\OP{F}}(P)$ denotes the number of forbidden pairs clustered into $P$, and
		
	\item
		$\#_{\OP{NF}}(P,P')$ denotes the number of non-forbidden pairs between $P$ and $P'$.
\end{itemize}

\begin{figure*}[h]
\centering
\includegraphics[scale=0.76]{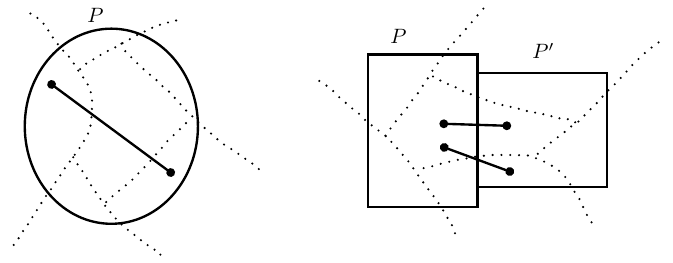}
\caption{Two types of disagreements to bound in this section -- (a) Forbidden pairs clustered into the same part $P$. (b) Non-forbidden pairs across different parts $P$ and $P'$. 
}
\end{figure*}

Also recall that for any cluster $P \in \MC{P}^{(t)}$, 
\begin{itemize}
	\item
		$\Delta^{(t)}(P)$ denotes the gluer set of $P$, $\Delta^{(t)}_+(P) := P \cap \Delta^{(t)}(P)$ is referred to as the core of $P$,

	\item
		$\OP{Ext}^{(t)}(P) := P \setminus \Delta^{(t)}_+(P)$ denotes the extended part of $P$, and
		
	\item
		$L^{(t)}_1(P)$ denotes the set of elements in the $2/3$-vicinity of $P' \cap \Delta^{(\ell(t,P))}_+(P)$ within $P'$ over all $\vphantom{\text{\Large T}_\text{\Large T}} P' \in \OP{Candi}^{(\ell(t,P))}( \Delta^{(\ell(t,P))}(P) )$, where $\ell(t,P)$ is the index of the top-most layer up to the $t$-th layer at which $P$ is newly-created.

\end{itemize}

\paragraph{Sketch of Lemma~\ref{lemma-number-of-disagreements-forbidden-pairs-within-P} -- Forbidden pairs within any $P$.}

To illustrate the ideas, we prove a weaker 
bound of $\frac{(2-\alpha)(1+\alpha)^2}{2 (1-\alpha)^2}$ for $\#_{\OP{F}}(P)$ in the following.
For $\frac{\beta(2-\alpha)(1+\alpha)^2}{2 (1-\alpha)^2}$ with $\beta := 0.8346$, we refer the readers to Section~\ref{sec-appendix-forbidden-pairs-within-P} in the appendix for the details.

\medskip

Let $P \in \MC{P}^{(t)}$ be a cluster.
Since $\OP{diam}^{(t)}(\Delta^{(t)}_+(P)) < 1/3$ and the distances are non-increasing bottom-up over the layers, forbidden pairs only occur between $\OP{Ext}^{(t)}(P)$ and $P$, i.e., no forbidden pairs reside within {$\Delta^{(t)}_+(P)$}.
Hence, we have
\begin{align}
\#_{\OP{F}}(P) \; 
\le \; |\OP{Ext}^{(t)}(P)| \cdot \left( \frac{|\OP{Ext}^{(t)}(P)|}{2} + |\Delta^{(t)}_+(P)| \right) \;
\le & \;\; \frac{1}{1-\alpha} \cdot |L^{(t)}_1(P)| \cdot \frac{2-\alpha}{2(1-\alpha)} \cdot |\Delta^{(t)}_+(P)|
\label{ieq-forbidden-same-P-orig-1-main} \\[2pt]
\le & \;\; \frac{2-\alpha}{2(1-\alpha)^2} \cdot \alpha \cdot |\Delta^{(t)}_+(P)|^2,
\label{ieq-forbidden-same-P-orig-2-main}
\end{align}
where in the last two inequalities we use the bounds from Lemma~\ref{lemma-ieq-ext-delta-plus-with-k}.

\smallskip

We have two cases to consider.
If $P$ is a newly-formed cluster at the $t$-th layer, then any pair between $\Delta^{(t)}_+(P)$ and $L^{(t)}_1(P)$ crosses different pre-clusters and is non-forbidden by the way $L^{(t)}_1(P)$ is defined.
Hence, these pairs are contained within $\OP{NFPrs}(\MC{Q}^{(t)}, P)$ and we have
$|L^{(t)}_1(P)| \cdot |\Delta^{(t)}_+(P)| \le \left| \OP{NFPrs}(\MC{Q}^{(t)}, P) \right|.$
It follows from~(\ref{ieq-forbidden-same-P-orig-1-main}) that
\begin{equation*}
\#_{\OP{F}}(P) \; \le \; \frac{2-\alpha}{2(1-\alpha)^2} \cdot \left| \OP{NFPrs}(\MC{Q}^{(t)}, P) \right|.
\end{equation*}

\begin{figure*}[tp]
\centering
\includegraphics[scale=0.8]{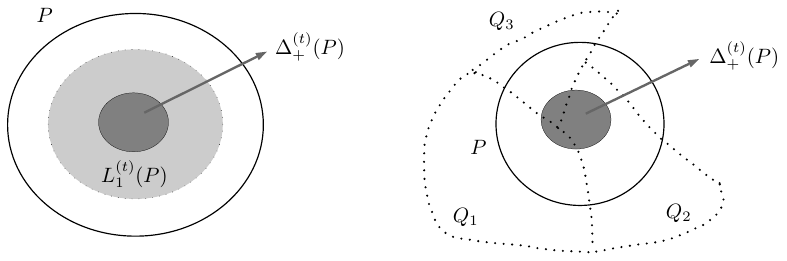}
\caption{
Two cases for the cluster $P$. (a) $P$ is newly-created at the $t$-th layer. In this case, 
the pairs between $L^{(t)}_1(P)$ and $\Delta^{(t)}_+(P)$ 
must be non-forbidden and reside between different pre-clusters within $P$.
(b) $P$ is created at a layer lower than $t$. In this case, $P \notin \OP{Candi}^{(t)}(Q)$ for any pre-cluster $Q$ that intersects $\Delta^{(t)}_+(P)$.
}
\label{fig-fbd-pair-counting-1}
\end{figure*}

\smallskip

Second, if $P$ is a previously-formed cluster at a lower layer, then consider the set of pre-clusters in $\MC{Q}^{(t)}$ that intersect the core set $\Delta^{(t)}_+(P)$.
Let $Q_1, \ldots, Q_k$ denote these pre-clusters and assume W.L.O.G. that $| Q_1 \cap \Delta^{(t)}_+(P) | = \max_{1\le j\le k} | Q_j \cap \Delta^{(t)}_+(P) |$.
Since $P \notin \OP{Candi}^{(t)}(Q_1)$, we have
\begin{equation}
B_1 \; := \; \left| \; \OP{Ball}^{(t)}_{<\frac{2}{3}} \left( \; P \cap Q_1, \; P \cap \overline{Q_1} \; \right) \; \right| \; \ge \; \alpha \cdot \left| P \cap Q_1 \right|.
\label{ieq-forbidden-same-P-concentration-condition-main}
\end{equation}

\noindent
We have two subcases to consider regarding the relative size of $| Q_j \cap \Delta^{(t)}_+(P) |$ for all $j$.

\begin{itemize}
\item
If $\sum_{2 \le j \le k} | Q_j \cap \Delta^{(t)}_+(P) | < \alpha \cdot | Q_1 \cap \Delta^{(t)}_+(P) |$, then
\begin{equation*}
|\Delta^{(t)}_+(P)|^2 \; \le \; (1+\alpha)^2 \cdot | Q_1 \cap \Delta^{(t)}_+(P) |^2 \; \le \; \frac{(1+\alpha)^2}{\alpha} \cdot | Q_1 \cap \Delta^{(t)}_+(P) | \cdot B_1
\end{equation*}
by Condition~(\ref{ieq-forbidden-same-P-concentration-condition-main}).
Since the pairs between $Q_1 \cap \Delta^{(t)}_+(P)$ and $\OP{Ball}^{(t)}_{<\frac{2}{3}} \left( \; P \cap Q_1, \; P \cap \overline{Q_1} \; \right)$ are non-forbidden, reside within $P$, and cross different pre-clusters, they are contained within $\OP{NFPrs}(\MC{Q}^{(t)}, P)$.
By~(\ref{ieq-forbidden-same-P-orig-2-main}) we have
\begin{equation*}
\#_{\OP{F}}(P) \; \le \; \frac{2-\alpha}{2(1-\alpha)^2} \cdot (1+\alpha)^2 \cdot \left| \OP{NFPrs}(\MC{Q}^{(t)}, P) \right|.
\end{equation*}

\item
If $\sum_{2 \le j \le k} | Q_j \cap \Delta^{(t)}_+(P) | \ge \alpha \cdot | Q_1 \cap \Delta^{(t)}_+(P) |$, since $\alpha \le 1/2$, it follows that $Q_1, \ldots, Q_k$ can be partitioned into two groups $\MC{G}_1$ and $\MC{G}_2$ such that\footnote{One of such ways is to consider $Q_j$ in non-ascending order of $| Q_j \cap \Delta^{(t)}_+(P) |$ for all $1\le j\le k$, and assign each $Q_j$ considered to the group that has a smaller intersection with $\Delta^{(t)}(P)$ in size.}
$$\alpha \cdot \sum_{Q \in \MC{G}_1} | Q \cap \Delta^{(t)}_+(P) | \;\; \le \;\; \sum_{Q \in \MC{G}_2} | Q \cap \Delta^{(t)}_+(P) | \;\; \le \;\; \sum_{Q \in \MC{G}_1} | Q \cap \Delta^{(t)}_+(P) |.$$

\noindent	
Define $G_1 := \sum_{Q \in \MC{G}_1} | Q \cap \Delta^{(t)}_+(P) |$ and $G_2 := \sum_{Q \in \MC{G}_2} | Q \cap \Delta^{(t)}_+(P) |$ for short.
We have
\begin{align*}
|\Delta^{(t)}_+(P)|^2 \;\; = \;\; ( G_1 + G_2 )^2 \;\; 
= \;\; & \left( \; \frac{G_1}{G_2} + \frac{G_2}{G_1} + 2 \; \right) \cdot G_1\cdot G_2 \notag \\[4pt]
\le \;\; & \left( \; \frac{1}{\alpha} + \alpha + 2 \; \right) \cdot G_1 \cdot G_2 \;\; = \;\; \frac{(1+\alpha)^2}{\alpha} \cdot G_1 \cdot G_2,
\end{align*}
where the last inequality follows since the function $f(x) = x+1/x$ attains its maximum value at $x = \alpha$ within the interval $[\alpha,1]$.
Since the pairs counted between $G_1$ and $G_2$ are contained within $\OP{NFPrs}(\MC{Q}^{(t)}, P)$, again we have
\begin{equation*}
\#_{\OP{F}}(P) \; \le \; \frac{2-\alpha}{2(1-\alpha)^2} \cdot (1+\alpha)^2 \cdot \left| \OP{NFPrs}(\MC{Q}^{(t)}, P) \right|.
\end{equation*}

\end{itemize}

\noindent
We provide the details for the improved bound $\frac{\beta(2-\alpha)(1+\alpha)^2}{2 (1-\alpha)^2}$ with $\beta := 0.8346$ in Section~\ref{sec-appendix-forbidden-pairs-within-P} in the appendix for further reference.

\paragraph{Proof of Lemma~\ref{lemma-number-of-disagreements-restate-non-forbidden-pairs-P-Pp} -- Non-forbidden pairs across $P$ and $P'$.}

This type of disagreements consists of two different types, namely, whether or not they reside within the same pre-cluster.

\begin{figure*}[h]
\centering
\includegraphics[scale=0.8]{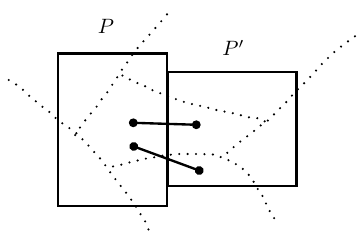}
\caption{
Two types of non-forbidden pairs across different parts $P$ and $P'$ -- Whether or not they reside within the same pre-cluster. 
}
\label{fig-disagreements-nfbd-1}
\end{figure*}

First, the number of non-forbidden pairs between $P$ and $P'$ that are not within the same pre-cluster is at most $\left|\OP{NFPrs}(\MC{Q}^{(t)}, P, P')\right|$,
where 
\begin{equation}
\OP{NFPrs}(\MC{Q}^{(t)}, P, P') := \{ \; \{i,j\} \in \OP{NFPrs}(\MC{Q}^{(t)}) \; \colon \; i \in P, \; j \in P' \; \}
\label{ieq-non-forbidden-diff-Q}
\end{equation}
denotes the set of non-forbidden pairs in $\OP{NFPrs}(\MC{Q}^{(t)})$ that resides between $P$ and $P'$.

\smallskip

Consider the set of non-forbidden pairs that resides between $P$ and $P'$ and that belongs to the same $Q$ for some $Q \in \MC{Q}^{(t)}$.
To count the number of such pairs, fix a $Q \in \MC{Q}^{(t)}$ with $P \cap Q \neq \emptyset$ and $P' \cap Q \neq \emptyset$.
By the design of Algorithm~\ref{algo-hier-clustering}, at most one of $P$ and $P'$ can be newly-created at this layer and have $Q$ being its gluer set. 

\smallskip

Without loss of generality, assume that $P$ is either previously-formed or newly-created with a gluer set other than $Q$.
In the following, for this $\{P, P'\}$ pair, we fix $P$ and count the number of non-forbidden pairs that reside in $Q$ and that have with one end in $P$ and the other end in $P'$.
We have two cases to consider.

\smallskip

If $Q \cap \Delta^{(t)}_+(P) = \emptyset$, then there exists $K^{(t)}_P(P\cap Q) \subseteq \OP{Ball}^{(t)}_{<2/3}(P \cap Q, P \setminus Q)$ such that 
$$|P \cap Q| \; \le \; \frac{\alpha}{1-\alpha} \cdot |K^{(t)}_P(P\cap Q)|$$
by Lemma~\ref{lemma-ieq-ext-delta-plus-with-k}.

Hence,
$$|P \cap Q| \cdot |P' \cap Q| \; \le \; \frac{\alpha}{1-\alpha} \cdot |K^{(t)}_P(P\cap Q)| \cdot |P'\cap Q|.$$
Note that, the pairs counted in the right-hand-side above reside between $P$ and $P'$.
Each of them has one end in $P' \cap Q$ and the other end in $P \cap Q'$ for some other pre-cluster $Q'$.
Moreover, they are non-forbidden.
It follows that
\begin{equation}
|P \cap Q| \cdot |P' \cap Q| \; \le \; \frac{\alpha}{1-\alpha} \cdot \left|\OP{NFPrs}(\MC{Q}^{(t)}, P, P', Q)\right|,
\label{ieq-non-forbidden-same-Q-1}
\end{equation}
where
$$\OP{NFPrs}(\MC{Q}^{(t)}, P, P', Q) := \{ \; \{i,j\} \in \OP{NFPrs}(\MC{Q}^{(t)}) \; \colon \; i \in P \setminus Q, \; j \in P' \cap Q \; \} $$ 
denotes the set of pairs in $\OP{NFPrs}(\MC{Q}^{(t)})$ that have one end in $P \setminus Q$ and the other in $P' \cap Q$.

\smallskip

For the second case, suppose that $Q \cap \Delta^{(t)}_+(P) \neq \emptyset$.
It follows that $P$ must be previously-formed. Furthermore, $P \notin \OP{Candi}^{(t)}(Q)$. Hence, we have
\begin{align}
|P \cap Q| \cdot |P' \cap Q| \; \le \;\; 
& \frac{1}{\alpha} \cdot \left| \; \OP{Ball}^{(t)}_{<\frac{2}{3}} \left( \; P \cap Q, \; P \cap \overline{Q} \; \right) \; \right|
\cdot |P' \cap Q| \notag \\[3pt]
\; \le \;\;
& \frac{1}{\alpha} \cdot \left|\OP{NFPrs}(\MC{Q}^{(t)}, P, P', Q)\right|.
\label{ieq-non-forbidden-same-Q-2}
\end{align}
Combining~(\ref{ieq-non-forbidden-same-Q-1}) and~(\ref{ieq-non-forbidden-same-Q-2}), it follows that 
$$\#_{\OP{NF}}(P,P',Q) \; \le \; \max \left\{ \; \frac{\alpha}{1-\alpha}, \; \frac{1}{\alpha} \; \right\} \cdot \left|\OP{NFPrs}(\MC{Q}^{(t)}, P, P', Q)\right|,$$
where $\#_{\OP{NF}}(P,P',Q)$ denotes the number of non-forbidden pairs that are between $P$ and $P'$ and that belong to $Q$.
Summing up over all $Q$ with $P \cap Q \neq \emptyset$ and $P' \cap Q \neq \emptyset$ and further taking~(\ref{ieq-non-forbidden-diff-Q}) into account, we obtain
\begin{equation}
\#_{\OP{NF}}(P,P') \; \le \; \left( \; \max \left\{ \; \frac{\alpha}{1-\alpha}, \; \frac{1}{\alpha} \; \right\} + 1 \; \right) \cdot \left|\OP{NFPrs}(\MC{Q}^{(t)}, P, P')\right|.
\label{ieq-non-forbidden-final}
\end{equation}

\smallskip

%% file: 5-bound-3.tex
\subsection{Average Distance of Non-Forbidden Cut Pairs}

\label{sec-proof-avg-distance-1-3}

Consider the execution of Algorithm~\ref{algo-precluster}.
Let $x$ be the input distance function.
Suppose that the algorithm picks a pair $(v,Q)$ with $v \in Q \in \MC{Q}$ and $\max_{u \in Q} x_{\{u,v\}} \ge 1/3$ in some iteration and let $(Q_1, Q_2)$ with $v \in Q_1$ be the pair returned by the procedure {\sc One-Third-Refine-Cut}.

\smallskip

Recall that we use $\OP{NFPrs}(Q_1, Q_2)$ to denote the set of non-forbidden pairs between $Q_1$ and $Q_2$.
For the ease of notation define
$$B_{1/3} \; := \; \OP{Ball}^{(x)}_{<1/3}(v, Q), \quad B_{1/2} \; := \; \OP{Ball}^{(x)}_{<1/2}(v, Q), \quad \text{and} \quad Q'_2 := Q_2 \cap B_{1/2}.$$
We prove the following lemma.

\begin{lemma}[Restate of Lemma~\ref{lemma-non-forbidden-pairs-separated-avg-cost-1-3}] \label{lemma-non-forbidden-pairs-separated-avg-cost-1-3-restate}
$$
\sum_{ \substack{ \{i,j\} \in \OP{NFPrs}(Q_1, Q_2), \\[2pt] i \in Q_1, \; j \in Q'_2 } } \left( \; \min\left\{ \; x_{\{v,j\}}, \; \frac{1}{3} \;\right\} \; - \; x_{\{v,i\}} \; \right) \;\; \ge \;\;
\frac{1}{6} \; \cdot \; \left| \left\{ \; \substack{ \{i,j\} \in \OP{NFPrs}(Q_1, Q_2), \\[2pt] i \in Q_1, \; j \in Q'_2 } \; \right\} \right|.
$$
\end{lemma}

\begin{proof}
For any $p,q \in B_{2/3}$, define $d(p,q) := | \min\{ x_{\{v, p\}}, 1/3\} - \min\{ x_{\{v,q\}}, 1/3\} | - 1/6$.
Since $Q_1 \subseteq B_{1/3}$, to prove the statement of this lemma, it suffices to prove that
\begin{equation}
\sum_{\{p,q\} \in \OP{NFPrs}(Q_1,Q'_2)} d(p,q) \; \ge \; 0.
\label{ieq-proof-pre-clustering-avg-frlp-1-3-1-0}
\end{equation}
From the setting of the procedure {\sc One-Third-Refine-Cut}, we have
$$(Q_1, Q'_2) \; \in \; 
\left\{ \; \begin{matrix}
\OP{Cut}_1 \; = \; \left( \; \{v\}, \; B_{1/2} \setminus \{v\} \; \right), \\[6pt]
\OP{Cut}_2 \; = \; \left( \; B_{1/3}, \; B_{1/2} \setminus B_{1/3} \; \right)
\end{matrix} \; \right\}.
$$

\noindent
Hence, to prove~(\ref{ieq-proof-pre-clustering-avg-frlp-1-3-1-0}), it suffices to prove that
\begin{equation}
W \; := \; \max_{1\le i\le 2} \left\{ \; \sum_{ \{p,q\} \in \OP{NFPrs}(\OP{Cut}_i) } d(p,q) \; \right\} \; \ge \; 0.
\label{ieq-proof-pre-clustering-avg-frlp-1-3-1}
\end{equation}
In the following we prove~(\ref{ieq-proof-pre-clustering-avg-frlp-1-3-1}).

\medskip

Let $k := |B_{1/3}|$ and $m := |B_{1/2} \setminus B_{1/3}|$.
For $\OP{Cut}_1$, 
we have
\begin{align}
\sum_{ \{p,q\} \in \OP{NFPrs}(\OP{Cut}_1) } d(p,q) \; 
= \; & \sum_{ q \in B_{1/3} } x_{\{v,q\}} \; + \; \frac{1}{3} \cdot |B_{1/2} \setminus B_{1/3}| \; - \; \frac{1}{6} \cdot \left( |B_{1/2}| - 1 \right) \notag \\[2pt]
= \; & \sum_{ q \in B_{1/3} } x_{\{v,q\}} \; + \; \frac{1}{6} \cdot |B_{1/2}| \; - \; \frac{1}{3} \cdot |B_{1/3}| \; + \; \frac{1}{6}
\label{ieq-proof-pre-clustering-avg-frlp-1-3-algo-condition} \\[2pt]
= \; & \sum_{ q \in B_{1/3} } x_{\{v,q\}} \; + \; \frac{1}{6} \cdot ( m-k+1).
\label{ieq-proof-pre-clustering-avg-frlp-1-3-2}
\end{align}
Note that the nonnegativity of~(\ref{ieq-proof-pre-clustering-avg-frlp-1-3-algo-condition}) is exactly tested by the procedure {\sc One-Third-Refine-Cut}.

\smallskip

For $\OP{Cut}_2$, observe that any $p \in B_{1/3}$ and $q \in B_{1/2}$ always forms a non-forbidden pair.
By a similar argument to the above, we have
\begin{align}
\sum_{ \{p,q\} \in \OP{NFPrs}(\OP{Cut}_2) } d(p,q) \; 
= \;\; &  \frac{1}{6} \cdot | \OP{NFPrs}(\OP{Cut}_2) | \; - \; m \cdot \sum_{ q \in B_{1/3} } x_{\{v,q\}}.
\label{ieq-proof-pre-clustering-avg-frlp-1-3-3}
\end{align}
From the definition of $W$ in~(\ref{ieq-proof-pre-clustering-avg-frlp-1-3-1}) combined with~(\ref{ieq-proof-pre-clustering-avg-frlp-1-3-2}) and~(\ref{ieq-proof-pre-clustering-avg-frlp-1-3-3}), we obtain
\begin{align}
W \; 
& \ge \;\; \frac{m}{m+1} \cdot \sum_{ \{p,q\} \in \OP{NFPrs}(\OP{Cut}_1) } d(p,q) \; + \; \frac{1}{m+1} \cdot \sum_{ \{p,q\} \in \OP{NFPrs}(\OP{Cut}_2) } d(p,q) \notag \\[6pt]
& = \;\; \frac{m}{m+1} \cdot \left( \; \sum_{ q \in B_{1/3} } x_{\{v,q\}} \; + \; \frac{1}{6} \cdot ( m-k +1) \; \right) + \; \frac{m}{m+1} \cdot \left( \; \frac{1}{6} \cdot k -  \sum_{ q \in B_{1/3} } x_{\{v,q\}} \; \right) \notag \\[6pt]
& = \;\; \frac{m}{6(m+1)}\cdot (m+1) \; \ge \; 0, \notag
\end{align}
where in the second last equality we use the fact that $| \OP{NFPrs}(\OP{Cut}_2) | = k\cdot m$.
\end{proof}

Recall that we define $Q'_2 := Q_2 \cap B_{1/2}$.
The following corollary, which is obtained by taking into accounts the pairs $\{i,j\}$ with $i \in Q_1$, $j \in Q_2 \setminus Q'_2$, summarizes the guarantee for the average distance of non-forbidden cut pairs.

\begin{corollary} \label{cor-non-forbidden-pairs-separated-avg-cost-1-3-overall}
\begin{align*}
\frac{1}{6} \cdot \left|  \OP{NFPrs}(Q_1, Q_2) \right|
\;\; \le \;\; 
& \sum_{ \substack{ \{i,j\} \in \OP{NFPrs}(Q_1, Q_2), \\[2pt] i \in Q_1, \; j \in Q'_2 } } \hspace{-4pt} \left( \; \min\left\{ \; x_{\{v,j\}}, \; \frac{1}{3} \;\right\} \; - \; x_{\{v,i\}} \; \right) \\[6pt]
& \hspace{0.8cm} + \;\; \sum_{ \substack{ \{i,j\} \in \OP{NFPrs}(Q_1, Q_2), \\[2pt] i \in Q_1, \; j \in Q_2 \setminus Q'_2, \\[2pt] \{i,j\} \in E^{(t)} } } x_{\{i,j\}}
\;\; + \;\; \frac{1}{6}\cdot \left| \left\{ \; \substack{ \{i,j\} \in \OP{NFPrs}(Q_1, Q_2), \\[2pt] i \in Q_1, \; j \in Q_2 \setminus Q'_2, \\[2pt] \{i,j\} \in NE^{(t)} } \; \right\} \right|
.
\end{align*}
\end{corollary}

\begin{proof}
Observe that, for any $i \in Q_1$, $j \in Q_2 \setminus Q'_2$, we have $x_{\{i,j\}} \ge x_{\{v,j\}} - x_{\{v,i\}} \ge 1/6$.
\end{proof}

The following lemma relates the number of non-forbidden pairs separated by $Q_1$ and $Q_2$ to the objective value of the input distance function in terms of the original input instance $(E^{(t)}, NE^{(t)})$.

\begin{lemma} \label{lemma-non-forbidden-pairs-separated-in-Q-1-3}
$$\sum_{ \substack{ \\[2pt] \{i,j\} \in \OP{NFPrs}(Q_1, Q_2), \\[2pt] \{i,j\} \in E^{(t)} } } \hspace{-16pt} x_{\{i,j\}} \; + \hspace{-4pt} \sum_{ \substack{ \\[2pt] \{i,j\} \in \OP{NFPrs}(Q_1, Q_2), \\[2pt] \{i,j\} \in NE^{(t)} } } \hspace{-10pt} \left(1-x_{\{i,j\}} \right)
\;\; \ge \;\; \frac{1}{6}\cdot \left| \left\{ \; \substack{ \{i,j\} \in \OP{NFPrs}(Q_1,Q_2), \\[2pt] \{i,j\} \in E^{(t)} } \; \right\} \right|.$$
\end{lemma}

\begin{proof}
Recall that $Q'_2 := Q_2 \cap B_{1/2}$.
First, we prove that
\begin{align}
& \sum_{ \substack{ \\[2pt] \{i,j\} \in \OP{NFPrs}(Q_1, Q_2), \\[2pt] \{i,j\} \in E^{(t)} } } \hspace{-16pt} x_{\{i,j\}} \; + \hspace{-4pt} \sum_{ \substack{ \\[2pt] \{i,j\} \in \OP{NFPrs}(Q_1, Q_2), \\[2pt] \{i,j\} \in NE^{(t)} } } \hspace{-10pt} \left(1-x_{\{i,j\}} \right) \; + \; \frac{1}{6}\cdot \left| \left\{ \; \substack{ \{i,j\} \in \OP{NFPrs}(Q_1,Q'_2), \\[2pt] \{i,j\} \in NE^{(t)} } \; \right\} \right| \notag \\[6pt]
& \hspace{1.8cm} \ge \;\; 
\sum_{ \substack{ \{i,j\} \in \OP{NFPrs}(Q_1, Q_2), \\[2pt] i \in Q_1, \; j \in Q'_2 } } \hspace{-4pt} \left( \; \min\left\{ \; x_{\{v,j\}}, \; \frac{1}{3} \;\right\} \; - \; x_{\{v,i\}} \; \right)
\; +  \sum_{ \substack{ \{i,j\} \in \OP{NFPrs}(Q_1, Q_2), \\[2pt] i \in Q_1, \; j \in Q_2 \setminus Q'_2, \\[2pt] \{i,j\} \in E^{(t)} } } \hspace{-8pt} x_{\{i,j\}}.
\label{ieq-non-forbidden-pairs-separated-in-Q-1-3-overall-1}
\end{align}

To prove~(\ref{ieq-non-forbidden-pairs-separated-in-Q-1-3-overall-1}), consider any $\{i,j\} \in \OP{NFPrs}(Q_1, Q_2)$ with $i \in Q_1$.
\begin{enumerate}
	\item
		If $\{i,j\}$ is an edge pair in $E^{(t)}$, then using the triangle inequality we have $x_{\{i,j\}} \ge x_{\{v,j\}} - x_{\{v,i\}}$ and hence $x_{\{i,j\}} \ge \min \left\{ x_{\{v,j\}}, \; \frac{1}{3} \right\} - x_{\{v,i\}} $. 

	\item
		If $\{i,j\}$ is a non-edge pair in $NE^{(t)}$ with $j \in Q'_2$, then 
		applying the setting and the triangle inequality we have $x_{\{i,j\}} \le x_{\{v,i\}} + x_{\{v,j\}} \le 5/6$, and hence
				$$\left( \; 1-x_{\{i,j\}} \; \right) \; + \; \frac{1}{6} \;\; \ge \;\; \frac{1}{3} \;\; \ge \;\; \min\left\{ \; x_{\{v,j\}}, \; \frac{1}{3} \;\right\} \; - \; x_{\{v,i\}}.$$
\end{enumerate}
The above compares the left-hand side of~(\ref{ieq-non-forbidden-pairs-separated-in-Q-1-3-overall-1}) with its right-hand side for all cases.
Hence, we have~(\ref{ieq-non-forbidden-pairs-separated-in-Q-1-3-overall-1}).
Adding
$\frac{1}{6}\cdot \left| \left\{ \; \substack{ \{i,j\} \in \OP{NFPrs}(Q_1, Q_2), \\[2pt] i \in Q_1, \; j \in Q_2 \setminus Q'_2, \\[2pt] \{i,j\} \in NE^{(t)} } \; \right\}^{\vphantom{\text{\large TT}}} \right|$ 
to both sides of~(\ref{ieq-non-forbidden-pairs-separated-in-Q-1-3-overall-1}) and applying Corollary~\ref{cor-non-forbidden-pairs-separated-avg-cost-1-3-overall}, it follows that %
$$
\sum_{ \substack{ \\[2pt] \{i,j\} \in \OP{NFPrs}(Q_1, Q_2), \\[2pt] \{i,j\} \in E^{(t)} } } \hspace{-20pt} x_{\{i,j\}} \; + \hspace{-6pt} \sum_{ \substack{ \\[2pt] \{i,j\} \in \OP{NFPrs}(Q_1, Q_2), \\[2pt] \{i,j\} \in NE^{(t)} } } \hspace{-22pt} \left(1-x_{\{i,j\}} \right) 
\; + \; \frac{1}{6} \cdot \left| \left\{ \substack{ \{i,j\} \in \OP{NFPrs}(Q_1, Q_2), \\[2pt] \{i,j\} \in NE^{(t)} } \right\} \right|
\; \ge \; 
\frac{1}{6} \cdot \left| \OP{NFPrs}(Q_1, Q_2) \right|,
$$
and this lemma follows.
\end{proof}

Since Lemma~\ref{lemma-non-forbidden-pairs-separated-in-Q-1-3} holds for every $(Q_1, Q_2)$ output by the procedure {\sc One-Third-Refine-Cut}, we have the following corollary for the pre-cluster $\MC{Q}$ output by Algorithm~\ref{algo-precluster}.

\begin{corollary}
\label{cor-non-forbidden-pairs-separated-orig-objective-restate}
$$\sum_{ \substack{ \\[2pt] \{i,j\} \in \OP{NFPrs}(\MC{Q}^{(t)}), \\[2pt] \{i,j\} \in E^{(t)} } } \hspace{-14pt} \tilde{x}^{(t)}_{\{i,j\}} \;\; + \hspace{-10pt} \sum_{ \substack{ \\[2pt] \{i,j\} \in \OP{NFPrs}(\MC{Q}^{(t)}), \\[2pt] \{i,j\} \in NE^{(t)} } } \hspace{-10pt} \left(1-\tilde{x}^{(t)}_{\{i,j\}} \right) \; + \; \frac{1}{6}\cdot \left| \left\{ \; \substack{ \{i,j\} \in \OP{NFPrs}(\MC{Q}^{(t)}), \\[2pt] \{i,j\} \in NE^{(t)} } \; \right\} \right|
\; \ge \; \frac{1}{6}\cdot \left| \OP{NFPrs}(\MC{Q}^{(t)}) \right|.$$
\end{corollary}

\newpage

%% file: 5-bound-4.tex
\subsection{Relating the Objectives}

\label{section-relating-objectives}

We prove the following key technical lemma regarding the weighted cardinality of non-forbidden non-edge pairs over the layers in any optimal LP-solution.

\begin{lemma}[Restate of Lemma~\ref{lemma-opt-sols-nonforbidden-pair-value}] \label{lemma-opt-sols-nonforbidden-pair-value-restate}
$$\sum_{1\le t\le \ell} \; \delta_t \cdot \left( \; \sum_{ \{u,v\} \in {E}^{(t)} } \tilde{x}^{(t)}_{u,v} \; + \; \hspace{-4pt} \sum_{ \{u,v\} \in NE^{(t)} } \left( \; 1 - \tilde{x}^{(t)}_{u,v} \; \right) \; \right) \; \ge \; \sum_{1\le t\le \ell} \; \delta_t \cdot \left| \; \OP{NFbdNE}^{(t)} \; \right|, $$
where we use $\OP{NFbdNE}^{(t)} := NE^{(t)} \setminus \OP{Fbd}^{(t)}$ to denote the set of non-forbidden non-edge pairs at the $t$-th layer.
\end{lemma}

Fix an optimal solution $\tilde{x}$ to~\ref{LP-HIER-CORR-CLUS}.
In the following we modify the constraints in~\ref{LP-HIER-CORR-CLUS} step by step, while keeping the invariant that $\tilde{x}$ remains an optimal solution to the working LP.

\smallskip

Let LP-(W) denote the current working LP, where LP-(W) is initially~\ref{LP-HIER-CORR-CLUS}.
For each variable $x^{(t)}_{\{u,v\}}$ such that $\vphantom{\text{\Large T}_{\text{\LARGE T}}} \tilde{x}^{(t)}_{\{u,v\}}=1$, replace all occurrences of $x^{(t)}_{\{u,v\}}$ in~LP-(W) with the constant $1$.
Note that the restriction of $\tilde{x}$ to the surviving variables continues to be an optimal solution for the working LP after this modification.
Continue this process until $\tilde{x}^{(t)}_{\{u,v\}} < 1$ for all surviving variables in $\tilde{x}$.

\smallskip

There are three types of constraints in~LP-(W) other than the nonnegativity constraints, namely, $\vphantom{\text{\Large}_{\text{\LARGE T}}} x^{(t)}_{\{u,p\}}+x^{(t)}_{\{p,v\}}\ge x^{(t)}_{\{u,v\}}$, $x^{(t)}_{\{u,v\}}\ge x^{(t+1)}_{\{u,v\}}$, and $x^{(t)}_{\{u,v\}}\le 1$.
We further modify the LP to remove \emph{redundant} constraints, which we describe in the following.
\begin{itemize}
    \item For each $1\le t \le \ell$ and $u,v,p \in V$, remove the constraint $x^{(t)}_{\{u,p\}}+x^{(t)}_{\{p,v\}}\ge x^{(t)}_{\{u,v\}}$ if at least one variable on the left-hand side was replaced with $1$.

    \item For each $1\le t\le \ell$ and $u,v \in V$, remove the constraint $x^{(t)}_{\{u,v\}}\ge x^{(t+1)}_{\{u,v\}}$ if at least one variable was replaced with $1$.
    
    \item For each $1\le t\le \ell$ and $u,v \in V$, remove $x^{(t)}_{\{u,v\}}\le 1$ if $x^{(t)}_{\{u,v\}}$ was replaced with $1$.
\end{itemize}

\noindent
Let $\OP{SV}^{(t)} := \{ \{u,v\} \colon \tilde{x}^{(t)}_{\{u,v\}} < 1 \}$ be the set of variables that survived in $\tilde{x}^*$ with the layer index $t$, $\tilde{x}^*$ be the restriction of $\tilde{x}$ to $\{\OP{V}^{(t)}\}_t$, and~\ref{LP-HIER-CORR-CLUS-Modified} be the LP obtained by the above procedure.

\smallskip

We have the following lemma.

\begin{lemma}
$\tilde{x}^*$ is an optimal solution for~\ref{LP-HIER-CORR-CLUS-Modified}.
\end{lemma}

\begin{proof}
We claim that removing the above constraints does not change the set of feasible solutions, and hence $\tilde{x}^*$ remains an optimal solution to the resulting LP. 

\smallskip

Consider the first type of constraints.
The removed constraints are in the form of $1+x^{(t)}_{\{p,v\}} \ge x^{(t)}_{\{u,v\}}$, $1+1 \ge x^{(t)}_{\{u,v\}}$, $1+x^{(t)}_{\{p,v\}} \ge 1$, or $1+1 \ge 1$. 
\begin{itemize}
	\item
		For $1+x^{(t)}_{\{p,v\}} \ge x^{(t)}_{\{u,v\}}$, where $\{p,v\}, \{u,v\} \in \OP{SV}^{(t)}$, the removed constraint is implied by $x^{(t)}_{\{p,v\}} \ge 0$ and $x^{(t)}_{\{u,v\}} \leq 1$, which are constraints that still exist in~\ref{LP-HIER-CORR-CLUS-Modified}.
		
	\item
		For $1+1 \ge x^{(t)}_{\{u,v\}}$, where $\{u,v\} \in \OP{SV}^{(t)}$, the removed constraint is implied by $x^{(t)}_{\{u,v\}} \leq 1$, a constraint still existing in~\ref{LP-HIER-CORR-CLUS-Modified}.
		
	\item
		For $1+x^{(t)}_{\{p,v\}} \ge 1$ with $\{p,v\} \in \OP{SV}^{(t)}$, again it is implied by $x^{(t)}_{\{p,v\}} \ge 0$, which exits in~\ref{LP-HIER-CORR-CLUS-Modified}.
\end{itemize}

\smallskip

Consider the second type of constraints, i.e., $x^{(t)}_{\{u,v\}}\ge x^{(t+1)}_{\{u,v\}}$. 
If $x^{(t+1)}_{\{u,v\}}$ was replaced with $1$, then $\tilde{x}^{(t)}_{\{u,v\}}=1$ and $\{u,v\} \notin \OP{SV}^{(t)}$. 
If only $x^{(t)}_{\{u,v\}}$ was replaced, then $1 \ge x^{(t+1)}_{\{u,v\}}$ is a constraint that persists in~\ref{LP-HIER-CORR-CLUS-Modified}.
Finally, for the third type of constraints, $x^{(t)}_{\{u,v\}}\le 1$, the claimed statement is trivial. 
This proves the lemma.
\end{proof}

\smallskip

\begin{figure*}[htp]
\centering
\fbox{
\begin{minipage}{.93\textwidth}
\begin{align}
\text{min} \; 
& \;\; \sum_{1\le t \le \ell} \delta_t \cdot \left( \; \left| E^{(t)} \setminus \OP{SV}^{(t)} \right| 
\; + \hspace{-6pt} \sum_{ \substack{ \{u,v\} \in E^{(t)}, \\[2pt] \{u,v\} \in \OP{SV}^{(t)} } } \hspace{-6pt} x^{(t)}_{\{u,v\}} 
\; + \hspace{-6pt} \sum_{ \substack{ \{u,v\} \in NE^{(t)}, \\[2pt] \{u,v\} \notin \OP{Fbd}^{(t)} } } \hspace{-6pt} ( 1-x^{(t)}_{\{u,v\}} ) \; \right)  & & \label{LP-HIER-CORR-CLUS-Modified} \tag*{LP-(**)} \\[8pt]
\text{s.t.} \; & \;\;\; x^{(t)}_{\{u,p\}} \; + \; x^{(t)}_{\{p,v\}} \; \ge \; x^{(t)}_{\{u,v\}},  & &  
\hspace{-7cm} \forall \; 1\le t\le \ell, \; \{u,p\}, \{p,v\}, \{u,v\} \in \OP{SV}^{(t)}, \notag \\[6pt]
& \;\;\; x^{(t)}_{\{u,p\}} \; + \; x^{(t)}_{\{p,v\}} \; \ge \; 1,  & & 
\hspace{-7cm}  \forall \; 1\le t\le \ell, \; \{u,p\}, \{p,v\} \in \OP{SV}^{(t)}, \; \{u,v\} \notin \OP{SV}^{(t)}, \notag \\[6pt]
& \;\;\; x^{(t+1)}_{\{u,v\}} \; \le \; x^{(t)}_{\{u,v\}},  & & 
\hspace{-7cm}  \forall \; 1\le t < \ell, \; \{u,v\} \in \OP{SV}^{(t)} \cap \OP{SV}^{(t+1)},  \notag \\[4pt]
& \;\;\; 0 \; \le \; x^{(t)}_{\{u,v\}} \; \le \; 1,  & & 
\hspace{-7cm}  \forall \; 1\le t \le \ell, \; \{u,v\} \in \OP{SV}^{(t)}.  \notag
\end{align}
\vspace{-6pt}
\end{minipage}\quad
}
\end{figure*}

\smallskip

Let us now consider the dual of~\ref{LP-HIER-CORR-CLUS-Modified}, which has an objective function of the following form
$$\max \; \sum_{1\le t\le \ell} \; \delta_t \cdot \left( \; \left| E^{(t)} \setminus \OP{SV}^{(t)} \right| 
\; + \; \left| \OP{NFbdNE}^{(t)} \right| 
\; + \hspace{-6pt} 
\sum_{ \{u,v\} \notin \OP{SV}^{(t)} } c^{(t)}_{\{u,v\}} \cdot \beta^{(t)}_{\{u,v\}} 
\; - \hspace{-4pt} 
\sum_{ \{u,v\} \in \OP{SV}^{(t)} } \eta^{(t)}_{\{u,v\}}
\; \right),$$
where $\{\beta^{(t)}_{\{u,v\}}\}_{1\le t\le \ell, \; \{u,v\} \notin \OP{SV}^{(t)}}$ and 
$\{\eta^{(t)}_{\{u,v\}}\}_{1\le t\le \ell, \; \{u,v\} \in \OP{SV}^{(t)}}$ 
are non-negative dual variables for the second set and the last set of constraints in~\ref{LP-HIER-CORR-CLUS-Modified}, respectively, and 
$$c^{(t)}_{\{u,v\}} \; := \; \left| \left\{ \; p \; \colon \; \{u,p\}, \{p,v\} \in \OP{SV}^{(t)} \right\} \right|.$$

\smallskip

Since $\tilde{x}^{*(t)}_{\{u,v\}} < 1$ for any $1\le t\le \ell$ and $\{u,v\} \in \OP{SV}^{(t)}$, the complementary slackness condition states that in any optimal dual solution with $\tilde{y}^*$ which contains $\eta^*$ as dual variables for the last set of constraints in~\ref{LP-HIER-CORR-CLUS-Modified}, we always have that
$$\eta^{*(t)}_{\{u,v\}} = 0 \enskip \text{ for any $1\le t\le \ell$ and $\{u,v\}\in \OP{SV}^{(t)}$.}$$
This implies that 
$\sum_{1\le t\le \ell} \; \delta_t \cdot | \OP{NFbdNE}^{(t)} | \le \OP{Val}(\tilde{y}^*) = \OP{Val}(\tilde{x}^*)$, 
where $\OP{Val}(\tilde{x}^*)$ and $\OP{Val}(\tilde{y}^*)$ denote the objective value of $\tilde{x}^*$ and $\tilde{y}^*$, and it follows that
\begin{align}
\sum_{1\le t\le \ell} \; \delta_t \cdot \left| \OP{NFbdNE}^{(t)} \right| \; \le \; \sum_{1\le t\le \ell} \; \delta_t \cdot \left( \; \sum_{ \{u,v\} \in {E}^{(t)} } \tilde{x}^{(t)}_{u,v} \; + \; \hspace{-4pt} \sum_{ \{u,v\} \in NE^{(t)} } \left( \; 1 - \tilde{x}^{(t)}_{u,v} \; \right) \; \right).\label{ieq-changing-objective}
\end{align}

\medskip

\subsection{Putting Things Together}
\label{sec-overall}

Now we are ready to prove Lemma~\ref{lemma-overall-approximation-ratio}.
Consider the statements of Lemma~\ref{lemma-number-of-disagreements-forbidden-pairs-within-P} and Lemma~\ref{lemma-number-of-disagreements-restate-non-forbidden-pairs-P-Pp}.
By the definition of $\OP{NFPrs}(\MC{Q}^{(t)}, P)$ and $\OP{NFPrs}(\MC{Q}^{(t)}, P, P')$, we have that
$$ \OP{NFPrs}(\MC{Q}^{(t)}) \; = \; \bigcup_{P \in \MC{P}^{(t)}} \OP{NFPrs}(\MC{Q}^{(t)}, P) 
\;\; \cup \; \bigcup_{P, P' \in \MC{P}^{(t)}, \; P \neq P'} \OP{NFPrs}(\MC{Q}^{(t)}, P, P').$$
Hence, the two lemmas give that

\begin{equation}
\sum_{P \in \MC{P}^{(t)}} \#_{\OP{F}}(P) \; + \sum_{ \substack{ P,P' \in \MC{P}^{(t)} \\[2pt] P\neq P' } } \#_{\OP{NF}}(P,P') 
\; \le \; c(\alpha) \cdot \left| \OP{NFPrs}(\MC{Q}^{(t)}) \right|,
\label{eq-overall-number-of-disagreements}
\end{equation}
where $c(\alpha) := \max \left\{ \frac{\beta(2-\alpha)(1+\alpha)^2}{2 (1-\alpha)^2}, \frac{1}{1-\alpha}, \frac{1+\alpha}{\alpha} \right\} \approx 3.5406$ for $\alpha := 0.3936$, and $\beta := 0.8346$.

\smallskip

Applying Inequality~(\ref{eq-overall-number-of-disagreements}) on Inequality~(\ref{ieq-overall-simp-version-1}), we have 
\begin{align}
{\#(\MC{P}^{(t)}) \; \le}
 \; &  \sum_{P \in \MC{P}^{(t)}} \#_{\OP{NFbdNE}}(P) + \sum_{P \in \MC{P}^{(t)}} \#_{\OP{F}}(P) + \sum_{ \substack{ P,P' \in \MC{P}^{(t)} \\[2pt] P\neq P' } } \#_{\OP{NF}}(P,P')  \notag \\[-3pt] 
\le \; & \;\; 
{| \OP{NFbdNE}^{(t)}|} 
\; + \; c(\alpha) \cdot \left| \OP{NFPrs}(\MC{Q}^{(t)}) \right|,
\label{ieq-proof-putting-things-overall-guarantee-1}
\end{align} 
where we use the fact that $\#_{\OP{NFbdNE}}(P)$ is the number of non-forbidden non-edge pairs clustered within $P$.
By Corollary~\ref{cor-non-forbidden-pairs-separated-orig-objective-restate}, the R.H.S. of~(\ref{ieq-proof-putting-things-overall-guarantee-1}) is upper-bounded by
$$
|\OP{NFbdNE}^{(t)}|
\; + \; c(\alpha) \cdot \left( \; 
|\OP{NFbdNE}^{(t)}| 
\; + \; 6 \cdot \left( \sum_{ \substack{ \\[2pt] \{i,j\} \in \OP{NFPrs}(\MC{Q}^{(t)}), \\[2pt] \{i,j\} \in E^{(t)} } } \hspace{-12pt} \tilde{x}^{(t)}_{\{i,j\}} \;\; + \hspace{-4pt} \sum_{ \substack{ \\[2pt] \{i,j\} \in \OP{NFPrs}(\MC{Q}^{(t)}), \\[2pt] \{i,j\} \in NE^{(t)} } } \hspace{-6pt} \left(1-\tilde{x}^{(t)}_{\{i,j\}} \right) \; \right) \right).
$$
Summing up the weighted disagreements over all layers $t$ with $1\le t\le \ell$ and apply~(\ref{ieq-changing-objective}), we obtain
\begin{align*}
& \sum_{1\le t \le \ell} \delta_t \cdot 
 \#(\MC{P}^{(t)}) 
\; \le \; \left( \; 7 c(\alpha) +1 \; \right) \cdot \sum_{1\le t\le \ell} \delta_t \cdot \left( \sum_{\{u,v\} \in E^{(t)}} \tilde{x}^{(t)}_{\{u,v\}} + \sum_{\{u,v\} \notin E^{(t)}} \left( 1-\tilde{x}^{(t)}_{\{u,v\}} \right) \right).
\end{align*}

\medskip

\newpage

%% file: 6-extension.tex
\section{Extension to Ultrametric Violation Distance}

Recall that given a set of pairwise measured distance for a set $V$ of elements, the goal of the ultrametric violation distance problem is to edit the minimum number of input distances so that there is a perfect fit to an ultrametric.
In~\cite{DBLP:journals/siamcomp/CohenAddadFLM25,DBLP:conf/soda/CharikarG24} the following formulation is introduced for this problem, where $t_{\{u,v\}}$ denotes the supposed layer at which $u$ and $v$ are separated in the ultrametric when a perfect fit for the given distances exists.

\begin{figure*}[h]
\centering
\fbox{
\begin{minipage}{.76\textwidth}
\begin{align}
\text{min} \; & \;\; \sum_{ \substack{ u, v \in V, \; u \neq v } }  \left( \; \left( \; 1-x^{(t_{\{u,v\}})}_{\{u,v\}} \; \right) \; + \; x^{(t_{\{u,v\}}+1)}_{\{u,v\}} \; \right) & & \label{LP-HIER-CORR-CLUS-L0} \tag*{LP-($L_0$)} \\[5pt]
\text{s.t.} \; & \;\;\; x^{(t)}_{\{u,v\}} \; \le \; x^{(t)}_{\{u,p\}} \; + \; x^{(t)}_{\{p,v\}},  & & \hspace{-1.4cm} \forall \; 1\le t\le \ell, \; u,v,p \in V, \notag \\[5pt]
& \;\;\; 0 \; \le \; x^{(t+1)}_{\{u,v\}} \; \le \; x^{(t)}_{\{u,v\}} \; \le \; 1,  & & \hspace{-1.4cm} \forall \; 1\le t < \ell, \; u,v \in V.  \notag
\end{align}
\vspace{-8pt}
\end{minipage}\quad
}
\caption{LP formulation for the Ultrametric Violation Distance.}
\label{fig-hier-corr-clus-natural-lp-L0}
\end{figure*}

\smallskip

As for the~\ref{LP-HIER-CORR-CLUS} for the hierarchical correlation clustering problem, we implicitly assume in the following that $x^{(t)}_{\{u,u\}}=0$ for all $u\in V$ holds in any feasible solution $x$ for~\ref{LP-HIER-CORR-CLUS-L0}.
Furthermore, we extend the definition such that $x^{(\ell+1)}_{\{u,v\}}:=0$ for any $u,v\in V$.

\begin{algorithm*}[htp]
\caption{Ultrametric-Violation-Distance$(\{\tilde{x}\}_{1\le t\le \ell})$} \label{algo-3-ultrametric-violation-distance}
\begin{algorithmic}[1]
\State Let $\MC{P}^{(\ell+1)} \gets \{V\}$.
\For {$t = \ell$ down to $1$}
	\State Let $\MC{P}^{(t)} \gets \MC{P}^{(t+1)}$.
	\While {$\OP{diam}^{(t)}(P) \ge 1/2$ for some $P \in \MC{P}^{(t)}$}
		\State Pick $P \in \MC{P}^{(t)}$ and $v \in P$ such that $\max_{u \in P} \tilde{x}^{(t)}_{\{u,v\}} \ge 1/2$.
		\State $\MC{P}' \gets$ {\sc One-Half-Refine-Cut}$(P, v, \tilde{x}^{(t)})$. \quad \texttt{// Compute a refined cut}
		\State Replace $P$ with the sets in $\MC{P}'$ in $\MC{P}^{(t)}$.
	\EndWhile
\EndFor
\State \Return $\{\MC{P}^{(t)}\}_{1\le t\le \ell}$.
\end{algorithmic}

\smallskip

\hrule

\smallskip

\begin{algorithmic}[1]
\Procedure{One-Half-Refine-Cut}{$P, v, x$}
	\If {Condition~(\ref{ieq-algo-3-ratio-4-cutting-condition}) is satisfied for $(P,v)$} 
		\State \Return $\{ \; \{v\}, \; P\setminus \{v\} \; \}$. \quad \texttt{// make $v$ a singleton}
	\Else
		\State \Return $\{ \; \OP{Ball}^{(x)}_{<1/2}(v,P), \;\; P \setminus \OP{Ball}^{(x)}_{< 1/2}(v,P) \; \}$. \quad \texttt{// cut at $1/2-\epsilon$}
	\EndIf
\EndProcedure
\end{algorithmic}
\end{algorithm*}

Let $\tilde{x}$ be an optimal solution to~\ref{LP-HIER-CORR-CLUS-L0}.
The algorithm begins with a big cluster $\MC{P}^{(\ell+1)} := \{ V\}$ and proceeds in a top-down manner.
For each iteration $t$ with $t = \ell, \ldots, 1$, the algorithm uses $\MC{P}^{(t)} := \MC{P}^{(t+1)}$ as the initial clustering and repeats until $\OP{diam}^{(t)}(P) < 1/2$ holds for all $P \in \MC{P}^{(t)}$.
If $P \in \MC{P}^{(t)}$ contains a pair $(u,v)$ with distance at least $1/2$, then the cutting procedure {\sc One-Half-Refine-Cut} is applied to separate this pair.

\smallskip

The procedure {\sc One-Half-Refine-Cut} takes as input a tuple $(P,v,x)$, where $P$ is a set, $v \in P$ is the pivot, and $x$ is a distance function, and tests the following condition.
If
\begin{equation}
\sum_{q \in \OP{Ball}^{(x)}_{<1/2}(v,P)} x_{\{v,q\}} \;\; \ge \;\; \frac{1}{2} \cdot |\OP{Ball}^{(x)}_{<1/2}(v,P)| \; - \; \frac{1}{4} \cdot |\OP{Ball}^{(x)}_{<3/4}(v,P)| \; - \; \frac{1}{4},
\label{ieq-algo-3-ratio-4-cutting-condition}
\end{equation}
then the algorithm makes $v$ a singleton cluster by replacing $P$ with $\{v\}$ and $P \setminus \{v\}$.
Otherwise, $P$ is replaced with $\OP{Ball}^{(x)}_{<1/2}(v,P)$ and $P \setminus \OP{Ball}^{(x)}_{< 1/2}(v,P)$.

\medskip

\subsection*{Approximation Guarantee}

\label{sec-L0-analysis-overview}

We prove the following theorem for the statement in Corollary~\ref{cor-5-approx-L0}.

\begin{theorem} \label{thm-overall-approximation-ratio-L0}
Let $\{\MC{P}^{(t)}\}_{1\le t\le \ell}$ be the output of Algorithm~\ref{algo-3-ultrametric-violation-distance} and $\hat{x}$ be the rounded integer distance function to which $\MC{P}^{(t)}$ corresponds.
We have
\begin{align*}
\sum_{ \substack{ u, v \in V, \\[2pt] u \neq v } }  \left( \; \left( \; 1-\hat{x}^{(t_{\{u,v\}})}_{\{u,v\}} \; \right) \; + \; \hat{x}^{(t_{\{u,v\}}+1)}_{\{u,v\}} \; \right) 
\;\; \le \;\; 
5 \cdot \sum_{ \substack{ u, v \in V, \\[2pt] u \neq v } }  \left( \; \left( \; 1-\tilde{x}^{(t_{\{u,v\}})}_{\{u,v\}} \; \right) \; + \; \tilde{x}^{(t_{\{u,v\}}+1)}_{\{u,v\}} \; \right),
\end{align*}
where $\tilde{x}$ is an optimal solution to~\ref{LP-HIER-CORR-CLUS-L0}.
\end{theorem}

\smallskip

Observe that each $\{u,v\}$ pair contributes exactly two items in the objective value of $\tilde{x}$, namely, $\tilde{x}^{(t_{\{u,v\}}+1)}_{\{u,v\}}$ and $1-\tilde{x}^{(t_{\{u,v\}})}_{\{u,v\}}$.
We consider $\{u,v\}$ an edge pair for all the layers above $t_{\{u,v\}}$ and a non-edge pair for the remaining layers.
In this regard, define 
$$E^{(t)} \; := \; \left\{ \; \{u,v\} \; \colon \; t_{\{u,v\}} < t \; \right\} 
\quad \text{and} \quad
NE^{(t)} \; := \; \left\{ \; \{u,v\} \; \colon \; t_{\{u,v\}} \ge t \; \right\}.$$
Define $\OP{Fbd} := \left\{ \{u,v\} \; \colon \; \tilde{x}^{(t_{\{u,v\}})}_{\{u,v\}} = 1 \right\}$ to be the set of \emph{forbidden pairs} and
$\OP{NFbd} := {\binom{V}{2}} \setminus \OP{Fbd}$.

\begin{figure*}[tp]
\centering
\includegraphics[scale=0.96]{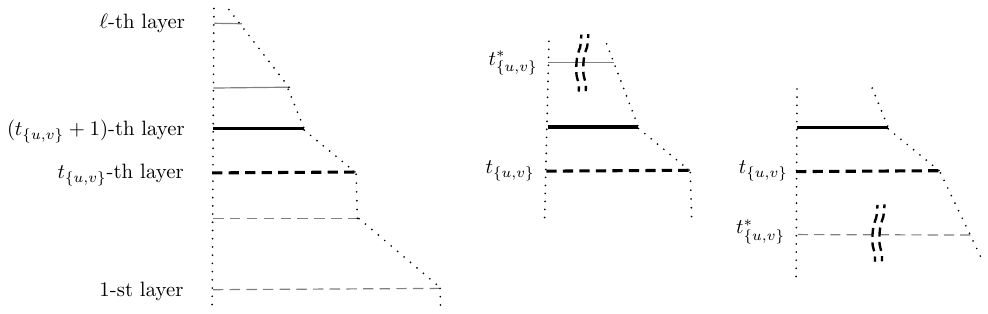}
\caption{
(a) We consider $\{u,v\}$ an edge pair for all the layers above the $t_{\{u,v\}}$-th layer and a non-edge pair for the remaining layers.
Moreover, $\{u,v\}$ contributes to the objective value only at the $(t_{\{u,v\}}+1)$-th and the $t_{\{u,v\}}$-th layers.
(b) Two types of disagreements for the $\{u,v\}$ pair, namely, $t^*_{\{u,v\}} > t_{\{u,v\}}$ or $t^*_{\{u,v\}} < t_{\{u,v\}}$.
}
\label{fig-two-layers-L0}
\end{figure*}

\smallskip

With exactly the same argument as in Lemma~\ref{lemma-opt-sols-nonforbidden-pair-value} (proved in Section~\ref{section-relating-objectives}), we have the following updated version of statement for~\ref{LP-HIER-CORR-CLUS-L0}. 

\begin{lemma} \label{lemma-opt-sols-nonforbidden-pair-value-L0}
Let $\tilde{x}$ be an optimal solution to~\ref{LP-HIER-CORR-CLUS-L0}.
We have
$$\sum_{ \substack{ u, v \in V, \\[2pt] u \neq v } }  \left( \; \left( \; 1-\tilde{x}^{(t_{\{u,v\}})}_{\{u,v\}} \; \right) \; + \; \tilde{x}^{(t_{\{u,v\}}+1)}_{\{u,v\}} \; \right) 
\;\; \ge \;\; 
\left| \; \OP{NFbd} \; \right|. $$
\end{lemma}

\medskip

Consider the execution of Algorithm~\ref{algo-3-ultrametric-violation-distance} and the calls the algorithm makes to the procedure {\sc One-Half-Refine-Cut}.
{
Let $k$ be the number of calls to the procedure and $\{ \hspace{1pt} (Q^{(i)}_1, Q^{(i)}_2) \hspace{1pt} \}_{1\le i\le k}$ be the set of pairs returned by the procedure upon these calls.
For each $(Q^{(i)}_1, Q^{(i)}_2)$, define 
$$\NFPrszero\left( Q^{(i)}_1, Q^{(i)}_2 \right) \; := \; \left\{ \; \{u,v\} \; \colon \; u \in Q^{(i)}_1, \; v \in Q^{(i)}_2, \; x^{(t_i)}_{\{u,v\}} < 1 \; \right\}$$
to be the set of pairs that are separated by $Q^{(i)}_1$ and $Q^{(i)}_2$ and that have distances strictly smaller than $1$ at the $t_i$-th layer, where we use $t_i$ to denote the layer at which $(Q^{(i)}_1, Q^{(i)}_2)$ is separated. 
We will refer these pairs to as \emph{non-extreme cut pairs}.
}

\medskip

For any $\{u,v\}$ pair, define 
$$\#_{\{u,v\}} := \left( \; 1-\hat{x}^{(t_{\{u,v\}})}_{\{u,v\}} \; \right) \; + \; \hat{x}^{(t_{\{u,v\}}+1)}_{\{u,v\}} 
\quad \text{and} \quad
\OP{Val}_{\{u,v\}} := \left( \; 1-\tilde{x}^{(t_{\{u,v\}})}_{\{u,v\}} \; \right) \; + \; \tilde{x}^{(t_{\{u,v\}}+1)}_{\{u,v\}}$$
to be the disagreement caused by the $\{u,v\}$ pair and the LP value the $\{u,v\}$ pair has, respectively.
To upper-bound $\#_{\{u,v\}}$, let $t^*_{\{u,v\}}$ be the top-most layer at which $\{u,v\}$ is separated for the first time in the hierarchy.
Define $t^*_{\{u,v\}}$ to be zero if $\{u,v\}$ is never separated in the hierarchy.

\medskip

It is clear that $\#_{\{u,v\}} = 1$ only when $t^*_{\{u,v\}} \neq t_{\{u,v\}}$.
Consider the following two cases.
\begin{itemize}
	\item
		$t^*_{\{u,v\}} = 0$.
		
		\smallskip
		
		In this case, $\{u,v\}$ is never separated.
		Then it follows from the design of Algorithm~\ref{algo-3-ultrametric-violation-distance} that
		$\tilde{x}^{(t_{\{u,v\}})}_{\{u,v\}} < \frac{1}{2}$.
		Hence,
		$$\#_{\{u,v\}} \; \le \; 2\cdot \left( \; 1-\tilde{x}^{(t_{\{u,v\}})}_{\{u,v\}} \; \right) \;\; \le \;\; 2\cdot \OP{Val}_{\{u,v\}}.$$
		
		\smallskip
		
	\item
		$0 < t^*_{\{u,v\}} \neq t_{\{u,v\}}$.
		
		\smallskip
		
		In this case, $\{u,v\}$ is separated by exactly one pair in $\MC{Q}$ and this happens at the $t^*_{\{u,v\}}$-th layer in the hierarchy.
		Denote this particular pair by $(Q^{(i)}_1, Q^{(i)}_2)$.

		\smallskip

		Further consider the following subcases.
		\begin{itemize}
			\item
				If $\{u,v\} \notin \NFPrszero(Q^{(i)}_1, Q^{(i)}_2)$ and $t^*_{\{u,v\}} > t_{\{u,v\}}$, then $\tilde{x}^{(t^*_{\{u,v\}})}_{\{u,v\}} = 1$, which implies that $\tilde{x}^{(t)}_{\{u,v\}} = 1$ for all $t \le t^*_{\{u,v\}}$, and hence
				\begin{equation*}
				\#_{\{u,v\}} \; = \; 1 \; = \; \tilde{x}^{(t_{\{u,v\}}+1)}_{\{u,v\}} \; = \; \OP{Val}_{\{u,v\}}.
				\end{equation*}
			
			\item
				If $\{u,v\} \notin \NFPrszero(Q^{(i)}_1, Q^{(i)}_2)$ and $t^*_{\{u,v\}} < t_{\{u,v\}}$, then it follows from the design of Algorithm~\ref{algo-3-ultrametric-violation-distance} that $\tilde{x}^{(t_{\{u,v\}})}_{\{u,v\}} < \frac{1}{2}$.
				Hence again 
				$$\#_{\{u,v\}} \; \le \; 2\cdot \left( \; 1-\tilde{x}^{(t_{\{u,v\}})}_{\{u,v\}} \; \right) \;\; \le \;\; 2\cdot \OP{Val}_{\{u,v\}}.$$
				
		\end{itemize}
\end{itemize}

\smallskip

\noindent
From the above three cases, we obtain that
\begin{align}
\sum_{u \neq v} \#_{\{u,v\}} \;\; 
= \;\; & \left| \left\{ \; \{u,v\} \; \colon \; 0 < t^*_{\{u,v\}} \neq t_{\{u,v\}} \; \right\} \right| \; + \; \left| \left\{ \; \{u,v\} \; \colon \; t^*_{\{u,v\}} = 0 \; \right\} \right| \notag \\[1pt]
\le \;\; & \sum_{1\le i \le k} \left| \NFPrszero( Q^{(i)}_1, Q^{(i)}_2 ) \right| \;\; + \;\; 2\cdot \hspace{-10pt} \sum_{ \substack{ \{u,v\} \colon t^*_{\{u,v\}} = 0 \text{ or } \\[2pt] \left( \; 0 \; < \; t^*_{\{u,v\}} \; \neq \; t_{\{u,v\}} \text{ and } \tilde{x}^{(t^*_{\{u,v\}})}_{\{u,v\}} = 1  \; \right) } } \hspace{-14pt} \OP{Val}_{\{u,v\}}.
\label{ieq-ultrametric-violation-distance-overall-1}
\end{align}

\noindent

The following lemma, which is the updated version of Corollary~\ref{cor-non-forbidden-pairs-separated-orig-objective-restate} for the Algorithm~\ref{algo-3-ultrametric-violation-distance}, bounds the number of pairs in $\OP{Fbd} \cap \NFPrszero(Q^{(i)}_1, Q^{(i)}_2)$ in terms of the average distance of the pairs in $\NFPrszero(Q^{(i)}_1, Q^{(i)}_2)$.
We provide the proof in Section~\ref{sec-proof-avg-distance-1-4} in the appendix for further reference.

\begin{lemma}[Section~\ref{sec-proof-avg-distance-1-4}]  \label{cor-non-forbidden-pairs-separated-orig-objective-L0}
Let $(Q_1, Q_2)$ be a pair returned by the procedure {\sc One-Half-Refine-Cut}.
We have that
$$\sum_{ \substack{ \\[2pt] \{i,j\} \in \NFPrszero(Q_1, Q_2) } } \hspace{-4pt} \OP{Val}_{\{i,j\}} \;\; + \;\;\; \frac{1}{4}\cdot \left| \left\{ \; \substack{ \{i,j\} \in \NFPrszero(Q_1, Q_2), \\[2pt] \{i,j\} \in \OP{NFbd} } \; \right\} \right|
\;\; \ge \;\; \frac{1}{4}\cdot \left| \NFPrszero(Q_1, Q_2) \right|.$$
\end{lemma}

\medskip

\noindent
Combining Lemma~\ref{cor-non-forbidden-pairs-separated-orig-objective-L0} with~(\ref{ieq-ultrametric-violation-distance-overall-1}) and Lemma~\ref{lemma-opt-sols-nonforbidden-pair-value-L0}, we obtain that
\begin{align}
\sum_{u \neq v} \#_{\{u,v\}} \;\; 
\le \;\;\; & 4\cdot \sum_{u \neq v} \OP{Val}_{\{u,v\}} \; + \; \left| \OP{NFbd} \right|
\;\; \le \;\; 5\cdot \sum_{u \neq v} \OP{Val}_{\{u,v\}}.
\notag
\end{align}
This proves Theorem~\ref{thm-overall-approximation-ratio-L0}.

\medskip
\medskip

%% file: 7-conclusion.tex
\section{Conclusion}

\label{sec-future-dir}

In this work, we present a new paradigm that advances the current understanding for hierarchical clustering in both conceptual and technical capacities.
A natural question following our results is whether the presented paradigm can be extended to other variations of hierarchical clustering problems with different objectives.
The technical problem boils down to the problem of finding cuts with prescribed properties regarding the average distances for the problem considered.

\smallskip

Another natural question is whether we can obtain better approximation result via improving the partitioning algorithm, e.g.,~\textsc{One-Half-Refine-Cut} in Algorithm~\ref{algo-3-ultrametric-violation-distance}.
The current partitioning algorithm can be interpreted as follows: sort the points by their distance from the pivot, and cut this sorted list either at distance $\epsilon$ or $1/2-\epsilon$. 
One could ask: \emph{what if we allow cutting to happen anywhere in the list?} We believe such an algorithm which partitions the ordered list of points into two consecutive sublists may be of  interest.

%% file: 8-appendix.tex
\begin{appendix}

\section{Proofs of Technical Lemmas}

\subsection{Lemma~\ref{lemma-number-of-disagreements-forbidden-pairs-within-P} -- Forbidden Pairs within any $P$.}

\label{sec-appendix-forbidden-pairs-within-P}

Let $P \in \MC{P}^{(t)}$ be a cluster and recall that
\begin{itemize}
	\item
		$\#_{\OP{F}}(P)$ denotes the number of forbidden pairs clustered into $P$,

	\item
		$\Delta^{(t)}(P)$ denotes the gluer set of $P$, $\Delta^{(t)}_+(P) := P \cap \Delta^{(t)}(P)$ is referred to as the core of $P$,

	\item
		$\OP{Ext}^{(t)}(P) := P \setminus \Delta^{(t)}_+(P)$ denotes the extended part of $P$, and
		
	\item
		$L^{(t)}_1(P)$ denotes the set of elements in the $2/3$-vicinity of $P' \cap \Delta^{(\ell(t,P))}_+(P)$ within $P'$ over all $\vphantom{\text{\Large T}_\text{\Large T}} P' \in \OP{Candi}^{(\ell(t,P))}( \Delta^{(\ell(t,P))}(P) )$, where $\ell(t,P)$ is the index of the top-most layer up to the $t$-th layer at which $P$ is newly-created.
		
\end{itemize}

\noindent
We prove the following lemma.

\begin{lemma}[Restate of Lemma~\ref{lemma-number-of-disagreements-forbidden-pairs-within-P}] \label{lemma-number-of-disagreements-forbidden-pairs-within-P-restate}
For $\alpha := 0.3936$ and any $P \in \MC{P}^{(t)}$, we have
$$\#_{\OP{F}}(P) \; \le \; \frac{(2-\alpha)(1+\alpha)^2}{2(1-\alpha)^2} \cdot \beta \cdot \left| \OP{NFPrs}(\MC{Q}^{(t)}, P) \right|,$$
where $\beta := 0.8346$ and $\OP{NFPrs}(\MC{Q}^{(t)}, P) := \left\{ \; \{i,j\} \in \OP{NFPrs}(\MC{Q}^{(t)}) \; \colon \; i,j \in P \; \right\} $
denotes the set of pairs in $\OP{NFPrs}(\MC{Q}^{(t)})$ residing within $P$.
\end{lemma}

\begin{proof}
Since $\OP{diam}^{(t)}(\Delta^{(t)}_+(P)) < 1/3$, forbidden pairs only occur between elements in $\OP{Ext}^{(t)}(P)$ and that in $P$.
Hence, we have
\begin{align}
\#_{\OP{F}}(P) \; 
\le & \;\; |\OP{Ext}^{(t)}(P)| \cdot \left( \frac{|\OP{Ext}^{(t)}(P)|}{2} + |\Delta^{(t)}_+(P)| \right) \notag \\[2pt]
\le & \;\; \frac{1}{1-\alpha} \cdot |L^{(t)}_1(P)| \cdot \frac{2-\alpha}{2(1-\alpha)} \cdot |\Delta^{(t)}_+(P)|,
\label{appendix-ieq-forbidden-same-P-orig-1} \\[4pt]
\le & \;\; \frac{2-\alpha}{2(1-\alpha)^2} \cdot \alpha \cdot |\Delta^{(t)}_+(P)|^2,
\label{appendix-ieq-forbidden-same-P-orig-2}
\end{align}
where we apply Lemma~\ref{lemma-ieq-ext-delta-plus-with-k} in the last two inequalities.

\smallskip

We have two cases to consider.
If $P$ is a newly-formed cluster at the $t$-th layer, then any pair between $\Delta^{(t)}_+(P)$ and $L^{(t)}_1(P)$ crosses different pre-clusters and is non-forbidden by the way $L^{(t)}_1(P)$ is defined.
Hence, these pairs are contained within $\OP{NFPrs}(\MC{Q}^{(t)}, P)$ and we have
$$|L^{(t)}_1(P)| \cdot |\Delta^{(t)}_+(P)| \; \le \; \left| \OP{NFPrs}(\MC{Q}^{(t)}, P) \right|.$$
Hence, from~(\ref{appendix-ieq-forbidden-same-P-orig-1}) we obtain
\begin{equation}
\#_{\OP{F}}(P) \; \le \; \frac{2-\alpha}{2(1-\alpha)^2} \cdot \left| \OP{NFPrs}(\MC{Q}^{(t)}, P) \right|.
\label{appendix-ieq-forbidden-same-P-final-1}
\end{equation}

\medskip

\begin{figure*}[h]
\centering
\includegraphics[scale=0.8]{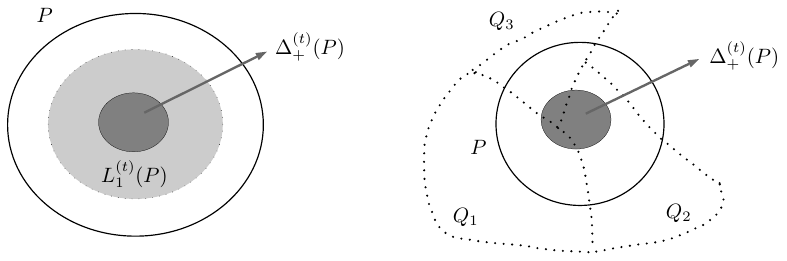}
\end{figure*}

\smallskip

If $P$ is a previously-formed cluster at a lower layer, then consider the set of pre-clusters in $\MC{Q}^{(t)}$ that intersect the core set $\Delta^{(t)}_+(P)$.
Let $Q_1, \ldots, Q_k$ denote these pre-clusters and assume W.L.O.G. that $| Q_1 \cap \Delta^{(t)}_+(P) | = \max_{1\le j\le k} | Q_j \cap \Delta^{(t)}_+(P) |$.
Since $P \notin \OP{Candi}^{(t)}(Q_1)$, by Step~6 of Algorithm~\ref{algo-hier-clustering}, we have
\begin{equation}
B_1 \; := \; \left| \; \OP{Ball}^{(t)}_{<\frac{2}{3}} \left( \; P \cap Q_1, \; P \cap \overline{Q_1} \; \right) \; \right| \; \ge \; \alpha \cdot \left| P \cap Q_1 \right|.
\label{appendix-ieq-forbidden-same-P-concentration-condition}
\end{equation}
We have two subcases to consider regarding the relative size of $| Q_j \cap \Delta^{(t)}_+(P) |$ for all $j$.

\paragraph{Case~(i) -- Imbalanced in Size.}
If $\sum_{2 \le j \le k} | Q_j \cap \Delta^{(t)}_+(P) | \; < \; \alpha \cdot | Q_1 \cap \Delta^{(t)}_+(P) |$, then
\begin{equation}
|\Delta^{(t)}_+(P)|^2 \;\; \le \;\; (1+\alpha)^2 \cdot | Q_1 \cap \Delta^{(t)}_+(P) |^2.
\label{appendix-ieq-forbidden-same-P-imbalance-1}
\end{equation}
To bound $| Q_1 \cap \Delta^{(t)}_+(P) |^2$, further consider two subcases regarding the size of $L^{(t)}_1(P)$ and $\Delta^{(t)}_+(P)$.
\begin{enumerate}
	\item
		If $|L^{(t)}_1(P)| \; \le \; \beta \cdot \alpha |\Delta^{(t)}_+(P)|$, then Inequality~(\ref{appendix-ieq-forbidden-same-P-orig-1}) yields a good bound.
		Combined with~(\ref{appendix-ieq-forbidden-same-P-imbalance-1}), we have
		\begin{align*}
		\#_{\OP{F}}(P) \; 
		\le \;\; & \frac{2-\alpha}{ \; 2(1-\alpha)^2 \; } \cdot \beta \cdot \alpha \cdot |\Delta^{(t)}_+(P)|^2 \\[4pt]
		\le \;\; & \frac{2-\alpha}{ \; 2(1-\alpha)^2 \; } \cdot (1+\alpha)^2 \cdot \beta \cdot \alpha \cdot | Q_1 \cap \Delta^{(t)}_+(P) |^2 \\[4pt]
		\le \;\; & \frac{2-\alpha}{ \; 2(1-\alpha)^2 \; } \cdot (1+\alpha)^2 \cdot \beta \cdot | Q_1 \cap \Delta^{(t)}_+(P) | \cdot B_1,
		\end{align*}
		where we use Condition~(\ref{appendix-ieq-forbidden-same-P-concentration-condition}) in the last inequality.
		Since the pairs between $Q_1 \cap \Delta^{(t)}_+(P)$ and $\vphantom{\text{\LARGE T}_\text{\LARGE T}} \OP{Ball}^{(t)}_{<2/3} \left( \; P \cap Q_1, \; P \cap \overline{Q_1} \; \right)$ are non-forbidden, reside within $P$, and cross different pre-clusters, it follows that
		\begin{equation}
		\#_{\OP{F}}(P) \; \le \; \frac{2-\alpha}{2(1-\alpha)^2} \cdot (1+\alpha)^2 \cdot \beta \cdot \left| \OP{NFPrs}(\MC{Q}^{(t)}, P) \right|.
		\label{appendix-ieq-forbidden-same-P-final-1-a}
		\end{equation}
		
	\item
		If $|L^{(t)}_1(P)| \; \ge \; \beta \cdot \alpha |\Delta^{(t)}_+(P)|$, then a decent number of elements exist in the $2/3$-vicinity of $\Delta^{(t)}_+(P)$.
		Define for short the following notations. 
		\begin{itemize}
			\item
				$\ell_1 := |Q_1 \cap L^{(t)}_1(P)|$ \; and \; $\ell_2 := \sum_{2\le j \le k} |Q_j \cap L^{(t)}_1(P)|$,
		
				\smallskip
		
			\item
				$\ell := |L^{(t)}_1(P)| - (\ell_1 + \ell_2)$,
		
				\smallskip
		
			\item
				$G_1 := | Q_1 \cap \Delta^{(t)}_+(P) |$ \; and \; $G_2 := \sum_{2\le j\le k} | Q_j \cap \Delta^{(t)}_+(P) |$.
		\end{itemize}
		Further consider two subcases regarding the relative size of $Q_1 \cap \OP{Ext}^{(t)}(P)$ and $L^{(t)}_1(P)$.
		
		\medskip
		
		If $\ell_1 \ge \; \eta \cdot | L^{(t)}_1(P) |$, where $\eta := \frac{\; 1-\beta \;}{ \; \alpha\beta^2 \; } \approx 0.6034$, then $Q_1$ contains a large number of elements in addition to those in $Q_1 \cap \Delta^{(t)}_+(P)$. 
		
		\smallskip
		
		In particular, we have 
		$\ell_1 \; \ge \; \eta \cdot | L^{(t)}_1(P) | \; \ge \; \alpha \beta \eta \cdot |\Delta^{(t)}_+(P)| \; \ge \; \alpha \beta \eta \cdot G_1$.
		Applying~Condition~(\ref{appendix-ieq-forbidden-same-P-concentration-condition}), we obtain
		\begin{align*}
		B_1 \; 
		\ge \;\; & \alpha \cdot \left| P \cap Q_1 \right| 
		\; {\ge} \;\; \alpha \cdot ( \ell_1 + G_1 ) \; 
		\ge \;\;  \alpha\cdot \left( 1+\alpha \beta \eta \right) \cdot G_1.
		\end{align*} 
		Following~(\ref{appendix-ieq-forbidden-same-P-imbalance-1}) and that $G_1 := | Q_1 \cap \Delta^{(t)}_+(P) |$, we obtain
		\begin{align*}
			|\Delta^{(t)}_+(P)|^2 \;\; \le \;\; (1+\alpha)^2 \cdot G_1^2 \;\; 
			\le \;\; & \frac{ \; (1+\alpha)^2 }{ \; \alpha\cdot ( 1+\alpha \beta \eta ) \; }  \cdot G_1 \cdot B_1 \; = \; \frac{ \; (1+\alpha)^2 \cdot \beta }{ \; \alpha \; }  \cdot G_1 \cdot B_1,
		\end{align*}
		where in the last equality we plug in the setting of $\eta$ to obtain that $\frac{1}{ \; 1+\alpha\beta\eta \; } = \beta$.
		Since $G_1 \cdot B_1 \le \left| \OP{NFPrs}(\MC{Q}^{(t)}, P) \right|$, 
		from~(\ref{appendix-ieq-forbidden-same-P-orig-2}) we have
		\begin{align}
		\#_{\OP{F}}(P) \; 
		\le \; & \; \frac{2-\alpha}{2(1-\alpha)^2} \cdot (1+\alpha)^2 \cdot \beta \cdot \left| \OP{NFPrs}(\MC{Q}^{(t)}, P) \right|.
		\label{appendix-ieq-forbidden-same-P-final-1-b}
		\end{align}
		
		\medskip
		
		If $\ell_1 \le \; \eta \cdot | L^{(t)}_1(P) |$, then a decent fraction of elements in $L^{(t)}_1(P)$ lies outside $Q_1$ and is ready to pair up with elements in $Q_1 \cap \Delta^{(t)}_+(P)$.
		We have
		\begin{equation}
		\ell_2 + \ell \; \ge \; \left( 1-\eta \right) \cdot | L^{(t)}_1(P) | \; \ge \; \alpha \cdot \beta \cdot \left(1-\eta \right) \cdot |\Delta^{(t)}_+(P)|.
		\label{appendix-ieq-forbidden-same-P-final-1-c-0}
		\end{equation}
		Let $\gamma := \frac{2\alpha\beta\cdot (1-\eta)  }{ \; 1+\alpha\beta\cdot(1-\eta) \; } \approx 0.2305$.
		We have
		\begin{align}
			|\Delta^{(t)}_+(P)|^2 \;\; 
			= \;\; & \left( \; G_1 \; + \; G_2 \; \right) \cdot |\Delta^{(t)}_+(P)| \notag \\[4pt]
			= \;\; & \gamma \cdot G_1 \cdot |\Delta^{(t)}_+(P)| \; + \; \left( \; (1-\gamma) \cdot G_1 \; + \; G_2 \; \right) \cdot \left( \; G_1 \; + \; G_2 \; \right) \notag \\[4pt]
			\le \;\; & \frac{ \gamma }{ \; \alpha \beta (1-\eta ) \; } \cdot G_1 \cdot \left( \ell_2 + \ell \right) \; + \; \frac{ \gamma }{ \; \alpha \beta (1-\eta ) \; } \cdot G_1 \cdot G_2 \notag \\[6pt]
			& \hspace{0.2cm} + \; ( \; 1-\gamma \; ) \cdot G_1^2 \; + \; \left( \; 2- \gamma - \frac{ \gamma }{ \; \alpha \beta  (1-\eta ) \; } \; \right) \cdot G_1 \cdot G_2 \; + \;  G_2^2,
			\label{appendix-ieq-forbidden-same-P-final-1-c-1}
		\end{align}
		where in the last inequality we apply Inequality~(\ref{appendix-ieq-forbidden-same-P-final-1-c-0}).
		Note that by the setting of $\gamma$, for the coefficient of $G_1 \cdot G_2$ in the above, we have
		\begin{equation*}
			2- \gamma - \frac{ \gamma }{ \; \alpha \beta (1-\eta ) \; } \;\; \ge \;\; 0.
		\end{equation*} 
		Hence, all the coefficients in the right-hand-side of~(\ref{appendix-ieq-forbidden-same-P-final-1-c-1}) are nonnegative, and it gives a valid upper-bound of $|\Delta^{(t)}_+(P)|^2$ in terms of pairs counted in $G_1 \cdot \left( \ell_2 + \ell + G_2 \right)$, \; $G_1^2$, \; $G_1\cdot G_2$, and $G_2^2$.
		Since $G_2 < \alpha \cdot G_1$, from~(\ref{appendix-ieq-forbidden-same-P-final-1-c-1}) we have
		\begin{align}
			|\Delta^{(t)}_+(P)|^2 \;\; 
			\le \;\; & \frac{ \gamma }{ \; \alpha \beta (1-\eta ) \; } \cdot G_1 \cdot \left( \ell_2 + \ell + G_2 \right)  \notag \\[6pt]
			& \hspace{0.1cm} \; + \; \left( \; (1-\gamma) \; + \; \alpha \cdot \left( \; 2- \gamma - \frac{ \gamma }{ \; \alpha \beta (1-\eta ) \; } \; \right) \; + \; \alpha^2 \; \right) \cdot G_1^2  \notag \\[8pt]
			\le \;\; & \frac{ \gamma }{ \; \alpha \beta (1-\eta ) \; } \cdot G_1 \cdot \left( \ell_2 + \ell + G_2 \right)  \notag \\[8pt]
			& \hspace{0.3cm} \; + \; \frac{1}{\alpha} \cdot \left( \; (1+\alpha)^2 \; - \; \gamma \cdot (1+\alpha) \; - \; \frac{ \gamma }{ \; \beta (1-\eta ) \; } \; \right) \cdot G_1 \cdot B_1,
			\label{appendix-ieq-forbidden-same-P-final-1-c-2}
		\end{align}
		where in the last inequality we apply Condition~(\ref{appendix-ieq-forbidden-same-P-concentration-condition}).
		Since the pairs counted in $G_1 \cdot \left( \ell_2 + \ell + G_2 \right)$ and $G_1 \cdot B_1$ are non-forbidden, it follows that
		$$ \max\left\{ \vphantom{\text{\LARGE T}} \; G_1 \cdot \left( \ell_2 + \ell + G_2 \right), \; G_1 \cdot B_1 \; \right\} \; \le \; \left| \OP{NFPrs}(\MC{Q}^{(t)}, P) \right|.$$
		Combining the above with~(\ref{appendix-ieq-forbidden-same-P-final-1-c-2}), we obtain
		\begin{equation}
			\#_{\OP{F}}(P) \; 
			\le \; \frac{2-\alpha}{2(1-\alpha)^2} \cdot (1+\alpha)^2 \cdot \left( \; 1 - \frac{\gamma}{1+\alpha} \; \right) \cdot \left| \OP{NFPrs}(\MC{Q}^{(t)}, P) \right|,
			\label{appendix-ieq-forbidden-same-P-final-1-c}
		\end{equation}
		where $\frac{\gamma}{ \; 1+\alpha \; } = \frac{1}{1+\alpha} \cdot \frac{ 2\alpha\beta^2 + 2\beta -2 }{ \alpha\beta^2 + 2\beta -1 }$ by plugging in the setting for $\eta$. 
		\medskip
		
\end{enumerate}

\paragraph{Case~(ii) -- Balanced in Size.}
If $\vphantom{\text{\LARGE T}_\text{\Large T}} \sum_{2 \le j \le k} | Q_j \cap \Delta^{(t)}_+(P) | \; \ge \; \alpha \cdot | Q_1 \cap \Delta^{(t)}_+(P) |$, 
since $\alpha \le 1/2$, it follows that $Q_1, \ldots, Q_k$ can be partitioned into two groups $\MC{G}_1$ and $\MC{G}_2$ such that\footnote{Note that, one way is to start with two empty groups and consider $Q_j$ in non-ascending order of $| Q_j \cap \Delta^{(t)}_+(P) |$ for all $1\le j\le k$.
For each $Q_j$ considered, assign it to the group that has a smaller intersection with $\Delta^{(t)}(P)$ in size.}
$$\alpha \cdot \sum_{Q \in \MC{G}_1} | Q \cap \Delta^{(t)}_+(P) | \;\; \le \;\; \sum_{Q \in \MC{G}_2} | Q \cap \Delta^{(t)}_+(P) | \;\; \le \;\; \sum_{Q \in \MC{G}_1} | Q \cap \Delta^{(t)}_+(P) |.$$
Define for short the following notations.
\begin{itemize}
	\item
		$G_1 := \sum_{Q \in \MC{G}_1} | Q \cap \Delta^{(t)}_+(P) |$, and
		
	\item
		$G_2 := \sum_{Q \in \MC{G}_2} | Q \cap \Delta^{(t)}_+(P) |$.
\end{itemize}
We have
\begin{align}
	|\Delta^{(t)}_+(P)|^2 \;\; = \;\; ( G_1 + G_2 )^2 \;\; 
	= \;\; & \left( \; \frac{G_1}{G_2} + \frac{G_2}{G_1} + 2 \; \right) \cdot G_1\cdot G_2 \notag \\[4pt]
	\le \;\; & \left( \; \frac{1}{\alpha} + \alpha + 2 \; \right) \cdot G_1 \cdot G_2 \;\; = \;\; \frac{(1+\alpha)^2}{\alpha} \cdot G_1 \cdot G_2,
	\label{appendix-ieq-forbidden-same-P-2-1}
\end{align}
where the last inequality follows since, within the interval $[\alpha,1]$, the function $f(x) = x+1/x$ attains its maximum value at $x = \alpha$.

\smallskip

Further consider two subcases regarding the relative size of $L^{(t)}_1(P)$ and $|\Delta^{(t)}_+(P)|$.
\begin{enumerate}
	\item
		If $|L^{(t)}_1(P)| \; \le \; \beta \cdot \alpha |\Delta^{(t)}_+(P)|$, then following Inequality~(\ref{appendix-ieq-forbidden-same-P-orig-1}) and~(\ref{appendix-ieq-forbidden-same-P-2-1}) we have
		\begin{align}
			\#_{\OP{F}}(P) \; 
			\le \;\; & \frac{2-\alpha}{ \; 2(1-\alpha)^2 \; } \cdot \beta \cdot \alpha \cdot |\Delta^{(t)}_+(P)|^2  \notag \\[4pt]
			\le \;\; & \frac{2-\alpha}{2(1-\alpha)^2} \cdot (1+\alpha)^2 \cdot \beta \cdot \left| \OP{NFPrs}(\MC{Q}^{(t)}, P) \right|.
			\label{appendix-ieq-forbidden-same-P-final-2-a}
		\end{align}
		
	\item
		If $|L^{(t)}_1(P)| \; \ge \; \beta \cdot \alpha |\Delta^{(t)}_+(P)|$, then a decent number of pairs exists between $L^{(t)}_1(P)$ and $\Delta^{(t)}_+(P)$.
		In this regard, define the following notations. 
		\begin{itemize}
			\item
				$\ell_1 := \sum_{Q \in \MC{G}_1} |Q \cap L^{(t)}_1(P)|$, \; $\ell_2 := \sum_{Q \in \MC{G}_2} |Q \cap L^{(t)}_1(P)|$, 
		
				\smallskip
			
			\item
				$L := |L^{(t)}_1(P)|$, \; and \; $\ell := L - (\ell_1 + \ell_2)$.
		
		\end{itemize}
		
		Further define $G := G_1\cdot \left( \ell_2 + \ell \right) \; + \; G_2 \cdot \left( \ell_1 + \ell \right)$ to count pairs between $L^{(t)}_1(P)$ and $\Delta^{(t)}_+(P)$.
		We have
		\begin{align*}
			G  \;\; 
			\ge \;\; & \left( \frac{\ell_2}{L} + \frac{\ell}{L} \right)\cdot G_1 \cdot L \; + \; \alpha \cdot \left( \frac{\ell_1}{L} + \frac{\ell}{L} \right)\cdot G_1 \cdot L \\[6pt]
			\ge \;\; & \alpha \cdot G_1 \cdot L \;\; \ge \;\; \alpha^2 \cdot \beta \cdot G_1 \cdot |\Delta^{(t)}_+(P)|, 
		\end{align*}
		which the second inequality follows from the fact that the previous R.H.S. attains its minimum value when $\ell_1 = L$ and $\ell_2 = \ell = 0$.
		Since $$G_1 \cdot |\Delta^{(t)}_+(P)| \;\; = \;\; G_1^2 + G_1\cdot G_2 \;\; \ge \;\; 2 \cdot G_1 \cdot G_2,$$
		we obtain that 
		\begin{align}
			G_1 \cdot G_2 \;\; 
			= \;\; & \; \zeta \cdot G_1 \cdot G_2 \; + \; (1-\zeta) \cdot G_1 \cdot G_2 \notag \\[4pt]
			\le \;\; & \frac{\zeta}{ \; 2 \alpha^2 \beta \; }\cdot G \; + \; \left(1-\zeta \right) \cdot G_1 \cdot G_2, 
			\label{appendix-ieq-forbidden-same-P-final-2-b-1}
		\end{align}
		where $\zeta := \frac{2\alpha^2\cdot \beta}{1+2\alpha^2\cdot \beta}$.
		Note that, the setting of $\zeta$ satisfies that $$\frac{\zeta}{ \; 2 \alpha^2 \beta \;} \; = \; 1-\zeta.$$
		Combining~(\ref{appendix-ieq-forbidden-same-P-final-2-b-1}) with~(\ref{appendix-ieq-forbidden-same-P-2-1}), we obtain
		\begin{align*}
			|\Delta^{(t)}_+(P)|^2 \;\; 
			\le \;\;  \frac{(1+\alpha)^2}{\alpha} \cdot 
			{ \frac{1}{ \; 1+ 2 \alpha^2 \beta \;} }
			\cdot \left( \; G \; + \; G_1 \cdot G_2 \; \right).
		\end{align*} 
		Since $G$ and $G_1\cdot G_2$ count two disjoint sets of non-forbidden pairs in $\OP{NFPrs}(\MC{Q}^{(t)}, P)$, it follows from~(\ref{appendix-ieq-forbidden-same-P-orig-2}) that
		\begin{align}
			\#_{\OP{F}}(P) \; 
			\le \;\; & \frac{2-\alpha}{2(1-\alpha)^2} \cdot (1+\alpha)^2 \cdot \frac{1}{ \; 2 \alpha^2 \beta \; } \cdot \left| \OP{NFPrs}(\MC{Q}^{(t)}, P) \right|.
			\label{appendix-ieq-forbidden-same-P-final-2-b}
		\end{align}
		
		\medskip
		
\end{enumerate}

Combining Inequalities~(\ref{appendix-ieq-forbidden-same-P-final-1-a}),~(\ref{appendix-ieq-forbidden-same-P-final-1-b}),~(\ref{appendix-ieq-forbidden-same-P-final-1-c}),~(\ref{appendix-ieq-forbidden-same-P-final-2-a}), and~(\ref{appendix-ieq-forbidden-same-P-final-2-b}), we obtain
\begin{align*}
\#_{\OP{F}}(P) \; 
\le \;\; & \frac{2-\alpha}{2(1-\alpha)^2} \cdot (1+\alpha)^2 \cdot W \cdot \left| \OP{NFPrs}(\MC{Q}^{(t)}, P) \right|,
\end{align*}
where 
{$$W := \max\left\{ \;\; \beta, \;\; 1 - \frac{ 2\alpha\beta^2 + 2\beta -2 }{ \; (1+\alpha) (\alpha\beta^2 + 2\beta -1) \; } , \;\; \frac{1}{ \; 1+2\alpha^2\beta \; } \;\; \right\},$$}
which has a value of $0.8346$ with the setting $\alpha := 0.3936$ and $\beta := 0.8346$.
Since $W = \beta$ and $(1+\alpha)^2 \beta \ge 1$, the statement of this lemma follows.
\end{proof}

\medskip
\medskip

\subsection{Lemma~\ref{cor-non-forbidden-pairs-separated-orig-objective-L0} -- Average Distance of \NonForbidden Cut Pairs}

\label{sec-proof-avg-distance-1-4}

Consider the procedure {\sc One-Half-Refine-Cut} with input tuple $(P,v,x)$, where $x$ is a distance function, $P$ is a set with $\OP{diam}^{(x)}(P) \ge 1/2$, and $v \in P$ is the pivot with $\max_{u \in P} x_{\{v,u\}} \ge 1/2$.

\begin{algorithm*}[h]

\smallskip

\begin{algorithmic}[1]
\Procedure{One-Half-Refine-Cut}{$P, v, x$}
	\If {Condition~(\ref{ieq-algo-3-ratio-4-cutting-condition}) is satisfied for $(P,v)$} 
		\State \Return $\{ \; \{v\}, \; P\setminus \{v\} \; \}$. \quad \texttt{// make $v$ a singleton}
	\Else
		\State \Return $\{ \; \OP{Ball}^{(x)}_{<1/2}(v,P), \;\; P \setminus \OP{Ball}^{(x)}_{< 1/2}(v,P) \; \}$. \quad \texttt{// cut at $1/2-\epsilon$}
	\EndIf
\EndProcedure
\end{algorithmic}
\end{algorithm*}

Suppose that the procedure is called at the $t$-th layer and $(Q_1, Q_2)$ with $v \in Q_1$ is the pair returned by the procedure {\sc One-Half-Refine-Cut}.
Recall that we use
$$\OP{Val}_{\{u,v\}} := \left( \; 1-\tilde{x}^{(t_{\{u,v\}})}_{\{u,v\}} \; \right) \; + \; \tilde{x}^{(t_{\{u,v\}}+1)}_{\{u,v\}}$$
to denote the objective value the pair $\{u,v\}$ possesses, $\NFPrszero(Q_1, Q_2)$ to denote the set of pairs with distances strictly smaller than $1$ between $Q_1$ and $Q_2$, and $\OP{NFbd} := {\binom{V}{2}} \setminus \OP{Fbd}$ to denote the set of $\{u,v\}$ pairs with $\tilde{x}^{(t_{\{u,v\}})}_{\{u,v\}} < 1$.

\medskip

In this section we prove the following lemma.

\begin{lemma}[Restate of Lemma~\ref{cor-non-forbidden-pairs-separated-orig-objective-L0}] \label{lemma-non-forbidden-pairs-separated-orig-objective-L0-proof}
$$\sum_{ \substack{ \\[2pt] \{i,j\} \in \NFPrszero(Q_1, Q_2) } } \hspace{-4pt} \OP{Val}_{\{i,j\}} \;\; + \;\;\; \frac{1}{4}\cdot \left| \left\{ \; \substack{ \{i,j\} \in \NFPrszero(Q_1, Q_2), \\[2pt] \{i,j\} \in \OP{NFbd} } \; \right\} \right|
\;\; \ge \;\; \frac{1}{4}\cdot \left| \NFPrszero(Q_1, Q_2) \right|.$$
\end{lemma}

\medskip

For the ease of notation define 
$$B_{1/4} \; := \; \OP{Ball}^{(x)}_{<1/4}(v, P), \quad B_{1/2} \; := \; \OP{Ball}^{(x)}_{<1/2}(v, P), \quad B_{3/4} \; := \; \OP{Ball}^{(x)}_{<3/4}(v, P), $$
and $Q'_2 := Q_2 \cap B_{3/4}$.
To prove Lemma~\ref{lemma-non-forbidden-pairs-separated-orig-objective-L0-proof}, first we bound the cardinality of $\NFPrszero(Q_1, Q_2)$ in terms of the average distance of the pairs it contains.
The following lemma is the updated version of Lemma~\ref{lemma-non-forbidden-pairs-separated-avg-cost-1-3} for the procedure {\sc One-Half-Refine-Cut}.
\begin{lemma}
\label{lemma-non-forbidden-pairs-separated-avg-cost-1-4}
$$
\sum_{ \substack{ \{i,j\} \in \NFPrszero(Q_1, Q_2), \\[2pt] i \in Q_1, \; j \in Q'_2 } } \left( \; \min\left\{ \; x_{\{v,j\}}, \; \frac{1}{2} \;\right\} \; - \; x_{\{v,i\}} \; \right) \;\; \ge \;\;
\frac{1}{4} \; \cdot \; \left| \left\{ \; \substack{ \{i,j\} \in \NFPrszero(Q_1, Q_2), \\[2pt] i \in Q_1, \; j \in Q'_2 } \; \right\} \right|.
$$
\end{lemma}

\begin{proof}
For any $p,q \in B_{3/4}$, define $d(p,q) := | \min\{ x_{\{v, p\}}, 1/2\} - \min\{ x_{\{v,q\}}, 1/2\} | - 1/4$.
Since $Q_1 \subseteq B_{1/2}$, to prove this lemma, it suffices to prove that
\begin{equation}
\sum_{\{p,q\} \in \NFPrszero(Q_1,Q'_2)} d(p,q) \; \ge \; 0.
\label{ieq-proof-pre-clustering-avg-frlp-1-2-1-0}
\end{equation}
From the design of the procedure {\sc One-Half-Refine-Cut}, we have
$$(Q_1, Q'_2) \; \in \; 
\left\{ \; \begin{matrix}
\OP{Cut}_1 \; = \; \left( \; \{v\}, \; B_{3/4} \setminus \{v\} \; \right), \\[6pt]
\OP{Cut}_2 \; = \; \left( \; B_{1/2}, \; B_{3/4} \setminus B_{1/2} \; \right)
\end{matrix} \; \right\}.
$$
Hence, to prove~(\ref{ieq-proof-pre-clustering-avg-frlp-1-2-1-0}), it suffices to prove that
\begin{equation}
W \; := \; \max_{1\le i\le 2} \left\{ \; \sum_{ \{p,q\} \in \NFPrszero(\OP{Cut}_i) } d(p,q) \; \right\} \; \ge \; 0.
\label{ieq-proof-pre-clustering-avg-frlp-1-4-1}
\end{equation}
In the following we prove~(\ref{ieq-proof-pre-clustering-avg-frlp-1-4-1}).

\smallskip

Let $k := |B_{1/4}|$, $\ell := |B_{1/2} \setminus B_{1/4}|$, and $m := |B_{3/4} \setminus B_{1/2}|$.
For $\OP{Cut}_1$, any $q \in B_{3/4} \setminus \{v\}$ always forms a \nonforbidden pair with $v$.
Hence, we have
\begin{align}
\sum_{ \{p,q\} \in \NFPrszero(\OP{Cut}_1) } d(p,q) \; 
= \; & \sum_{ q \in B_{1/2} } x_{\{v,q\}} \; + \; \frac{1}{2} \cdot |B_{3/4} \setminus B_{1/2}| \; - \; \frac{1}{4} \cdot \left( |B_{3/4}| - 1 \right) \notag \\[2pt]
= \; & \sum_{ q \in B_{1/2} } x_{\{v,q\}} \; + \; \frac{1}{4} \cdot |B_{3/4}| \; - \; \frac{1}{2} \cdot |B_{1/2}| \; + \; \frac{1}{4}
\label{ieq-proof-pre-clustering-avg-frlp-algo-condition-L0} \\[2pt]
= \; & \sum_{ q \in B_{1/2} } x_{\{v,q\}} \; + \; \frac{1}{4} \cdot ( m-k-\ell +1).
\label{ieq-proof-pre-clustering-avg-frlp-2-L0}
\end{align}
Note that {the nonnegativity of}~(\ref{ieq-proof-pre-clustering-avg-frlp-algo-condition-L0}) is exactly the condition tested by the procedure {\sc One-Half-Refine-Cut}.

\smallskip

For $\OP{Cut}_2$, observe that any $p \in B_{1/4}$ and $q \in B_{3/4}$ always forms a \nonforbidden pair.
For any $p \in B_{1/2} \setminus B_{1/4}$, let $N(p)$ denote the number of elements in $B_{3/4} \setminus B_{1/2}$ that forms a \nonforbidden pair with $p$.
It follows that
\begin{align}
\sum_{ \{p,q\} \in \NFPrszero(\OP{Cut}_2) } d(p,q) \; 
= \;\; &  \frac{1}{4} \cdot | \NFPrszero(\OP{Cut}_2) | \; - \; m \cdot \sum_{ q \in B_{1/4} } x_{\{v,q\}} \; - \hspace{-1pt} \sum_{q \in B_{1/2} \setminus B_{1/4}} \hspace{-4pt} N(q) \cdot x_{\{v,q\}} .
\label{ieq-proof-pre-clustering-avg-frlp-3-L0}
\end{align}
From the definition of $W$ in~(\ref{ieq-proof-pre-clustering-avg-frlp-1-4-1}) with~(\ref{ieq-proof-pre-clustering-avg-frlp-2-L0}) and~(\ref{ieq-proof-pre-clustering-avg-frlp-3-L0}), we obtain
\begin{align}
W \; 
& \ge \;\; \frac{m}{m+1} \cdot \sum_{ \{p,q\} \in \NFPrszero(\OP{Cut}_1) } d(p,q) \; + \; \frac{1}{m+1} \cdot \sum_{ \{p,q\} \in \NFPrszero(\OP{Cut}_2) } d(p,q) \notag \\[4pt]
& = \;\; \frac{m}{m+1} \cdot \left( \; \sum_{ q \in B_{1/2} } x_{\{v,q\}} \; + \; \frac{1}{4} \cdot ( m-k-\ell +1) \; \right) \notag \\[2pt]
& \hspace{1.2cm} + \; \frac{1}{m+1} \cdot \left( \; \frac{1}{4} \cdot | \NFPrszero(\OP{Cut}_2) | \; - \; m \cdot \sum_{ q \in B_{1/4} } x_{\{v,q\}} \; - \hspace{-1pt} \sum_{q \in B_{1/2} \setminus B_{1/4}} N(q) \cdot x_{\{v,q\}} \; \right). \notag
\end{align}
Further plugging in $| \NFPrszero(\OP{Cut}_2) | \; = \; m\cdot k + \sum_{q \in B_{1/2} \setminus B_{1/4}} N(q)$, we obtain
\begin{align}
W \; 
& \ge \;\; \frac{1}{m+1} \cdot \left( \; \sum_{q \in B_{1/2} \setminus B_{1/4}} \hspace{-4pt} ( m - N(q)) \cdot x_{\{v,q\}} \; + \; \frac{1}{4} \cdot m( m - \ell + 1 ) \; + \; \frac{1}{4} \cdot \hspace{-4pt} \sum_{q \in B_{1/2} \setminus B_{1/4}} \hspace{-10pt} N(q) \; \right) \notag \\[2pt]
& \ge \;\; \frac{m}{4(m+1)}\cdot (m+1)
\; \ge \; 0, \notag
\end{align}
where in the second last inequality we use the fact that $x_{\{v,q\}} \ge 1/4$ for any $q \in B_{1/2} \setminus B_{1/4}$.
\end{proof}

\smallskip

Recall that $t$ is the layer at which the procedure {\sc One-Half-Refine-Cut} is called and the pair $(Q_1, Q_2)$ with $v \in Q_1$ is separated.
Also recall that $E^{(t)}$ and $NE^{(t)}$ denote the set of edge pairs and the set of non-edge pairs at the $t$-th layer.

\smallskip

We have the following lemma.

\begin{lemma} \label{lemma-non-forbidden-pairs-separated-in-Q-1-4-intermediate-objective}
\begin{align*}
& \sum_{ \substack{ \\[2pt] \{i,j\} \in \NFPrszero(Q_1, Q_2), \\[2pt] \{i,j\} \in E^{(t)} } } \hspace{-16pt} x_{\{i,j\}} \; + \hspace{-4pt} \sum_{ \substack{ \\[2pt] \{i,j\} \in \NFPrszero(Q_1, Q_2), \\[2pt] \{i,j\} \in NE^{(t)} } } \hspace{-10pt} \left(1-x_{\{i,j\}} \right) \; + \; \frac{1}{4}\cdot \left| \left\{ \; \substack{ \{i,j\} \in \NFPrszero(Q_1,Q'_2), \\[2pt] \{i,j\} \in NE^{(t)} } \; \right\} \right| \notag \\[6pt]
& \hspace{1.8cm} \ge \;\; 
\sum_{ \substack{ \{i,j\} \in \NFPrszero(Q_1, Q_2), \\[2pt] i \in Q_1, \; j \in Q'_2 } } \hspace{-4pt} \left( \; \min\left\{ \; x_{\{v,j\}}, \; \frac{1}{2} \;\right\} \; - \; x_{\{v,i\}} \; \right)
\; +  \sum_{ \substack{ \{i,j\} \in \NFPrszero(Q_1, Q_2), \\[2pt] i \in Q_1, \; j \in Q_2 \setminus Q'_2, \\[2pt] \{i,j\} \in E^{(t)} } } \hspace{-8pt} x_{\{i,j\}}.
\end{align*}
\end{lemma}

\begin{proof}
To prove this lemma, we compare both sides of the inequality for each 
$\{i,j\} \in \NFPrszero(Q_1, Q_2)$ with $i \in Q_1$.
\begin{enumerate}
	\item
		If $\{i,j\}$ is an edge pair in $E^{(t)}$, then {using the triangle inequality, we have $x_{\{i,j\}} \ge x_{\{v,j\}} - x_{\{v,i\}}$} and hence $x_{\{i,j\}} \ge \min \left\{ x_{\{v,j\}}, \; \frac{1}{2} \right\} - x_{\{v,i\}} $.
		
	\item
		If $\{i,j\}$ is a non-edge pair in $NE^{(t)}$ with $j \in Q'_2$, then further consider the following subcases.
		
		\begin{enumerate}
			\item
				If $Q_1$ is a singleton-cluster, then it follows that $x_{\{i,j\}} \le 3/4$ and 
				$$\left( \; 1-x_{\{i,j\}} \; \right) \; + \; \frac{1}{4} \;\; \ge \;\; \frac{1}{2} \;\; \ge \;\; \min\left\{ \; x_{\{v,j\}}, \; \frac{1}{2} \;\right\} \; - \; x_{\{v,i\}}.$$
				
			\item
				If $Q_1 = B_{1/2}$, then $j \in B_{3/4} \setminus B_{1/2}$. 
				
				\smallskip
				
				\begin{enumerate}
					\item
				If $i \in B_{1/4}$, then {the} triangle inequality implies that
				$$1 - x_{\{i,j\}} \; \ge \; 1 - x_{\{v,i\}} - x_{\{v,j\}} \; \ge \; \frac{1}{4} - x_{\{v,i\}}.$$
				On the other hand, $\min\{ x_{\{v,j\}}, \frac{1}{2} \} - x_{\{v,i\}} = \frac{1}{2} - x_{\{v,i\}}$.
				Hence, 
				$$\left( \; 1-x_{\{i,j\}} \; \right) \; + \; \frac{1}{4} \;\; \ge \;\; \min\left\{ \; x_{\{v,j\}}, \; \frac{1}{2} \;\right\} \; - \; x_{\{v,i\}}.$$
				
					\item
						If $i \in B_{1/2} \setminus B_{1/4}$, then
						$${\left( \; 1-x_{\{i,j\}} \; \right) \; + \; \frac{1}{4} \;\; \ge \;\;} \frac{1}{4} \;\; = \;\; \frac{1}{2} \; - \; \frac{1}{4} \;\; \ge \;\; \min\left\{ \; x_{\{v,j\}}, \frac{1}{2} \; \right\} \; - \; x_{\{v,i\}}.$$
				\end{enumerate}				
		\end{enumerate}
\end{enumerate}
We have compared all $\{i,j\}$ in the above case arguments.
This proves the lemma.
\end{proof}

\medskip

In the following we prove Lemma~\ref{lemma-non-forbidden-pairs-separated-orig-objective-L0-proof}.

\begin{proof}[Proof of Lemma~\ref{lemma-non-forbidden-pairs-separated-orig-objective-L0-proof}]
We have that
\begin{equation*}
\sum_{\{i,j\} \in \NFPrszero(Q_1, Q_2)} \OP{Val}_{\{i,j\}} \; \ge \; \sum_{ \substack{ \\[2pt] \{i,j\} \in \NFPrszero(Q_1, Q_2), \\[2pt] \{i,j\} \in E^{(t)} } } \hspace{-16pt} x_{\{i,j\}} \; + \hspace{-4pt} \sum_{ \substack{ \\[2pt] \{i,j\} \in \NFPrszero(Q_1, Q_2), \\[2pt] \{i,j\} \in NE^{(t)} } } \hspace{-10pt} \left(1-x_{\{i,j\}} \right)
\end{equation*}
by the definition of $E^{(t)}$, $NE^{(t)}$, and the non-decreasing property of $\tilde{x}_{\{i,j\}}$ over the layers.
Combining the above statement with Lemma~\ref{lemma-non-forbidden-pairs-separated-in-Q-1-4-intermediate-objective}, we obtain that
\begin{align}
& \sum_{ \substack{ \\[2pt] \{i,j\} \in \NFPrszero(Q_1, Q_2) } } \hspace{-16pt} \OP{Val}_{\{i,j\}} 
\; + \;\; 
\frac{1}{4} \cdot \left| \left\{ \substack{ \{i,j\} \in \NFPrszero(Q_1, Q'_2), \\[2pt] \{i,j\} \in NE^{(t)} } \right\} \right|  \notag \\[4pt]
& \hspace{0.8cm} \ge \;\; 
\sum_{ \substack{ \{i,j\} \in \NFPrszero(Q_1, Q_2), \\[2pt] i \in Q_1, \; j \in Q'_2 } } \hspace{-4pt} \left( \; \min\left\{ \; x_{\{v,j\}}, \; \frac{1}{2} \;\right\} \; - \; x_{\{v,i\}} \; \right)
\; +  \sum_{ \substack{ \{i,j\} \in \NFPrszero(Q_1, Q_2), \\[2pt] i \in Q_1, \; j \in Q_2 \setminus Q'_2, \\[2pt] \{i,j\} \in E^{(t)} } } \hspace{-8pt} x_{\{i,j\}}  \notag \\[4pt]
& \hspace{0.8cm} \ge \;\; 
\sum_{ \substack{ \{i,j\} \in \NFPrszero(Q_1, Q_2), \\[2pt] i \in Q_1, \; j \in Q'_2 } } \hspace{-4pt} \left( \; \min\left\{ \; x_{\{v,j\}}, \; \frac{1}{2} \;\right\} \; - \; x_{\{v,i\}} \; \right)
\; + \;\; 
\frac{1}{4} \cdot \left| \left\{ \substack{ \{i,j\} \in \NFPrszero(Q_1, Q_2), \\[2pt] i \in Q_1, \; j \in Q_2 \setminus Q'_2, \\[2pt] \{i,j\} \in E^{(t)} } \right\} \right|,
\label{ieq-lemma-non-forbidden-pairs-separated-orig-objective-L0-proof-2}
\end{align}
where in the last inequality we use the fact that $x_{\{i,j\}} \ge 1/4$ for any $\{i,j\} \in \NFPrszero(Q_1, Q_2)$ with $i \in Q_1, \; j \in Q_2 \setminus Q'_2$ by the design of the procedure {\sc One-Half-Refine-Cut}.

\smallskip

Adding
$\frac{1}{4}\cdot \left| \left\{ \; \substack{ \{i,j\} \in \NFPrszero(Q_1, Q_2), \\[2pt] i \in Q_1, \; j \in Q_2 \setminus Q'_2, \\[2pt] \{i,j\} \in NE^{(t)} } \; \right\}^{\vphantom{\text{\large TT}}} \right|$ {to} 
both sides of~(\ref{ieq-lemma-non-forbidden-pairs-separated-orig-objective-L0-proof-2}) and combining it with Lemma~\ref{lemma-non-forbidden-pairs-separated-avg-cost-1-4}, we obtain that
$$
\sum_{ \substack{ \\[2pt] \{i,j\} \in \NFPrszero(Q_1, Q_2) } } \hspace{-16pt} \OP{Val}_{\{i,j\}} \; + \;\; \frac{1}{4} \cdot \left| \left\{ \substack{ \{i,j\} \in \NFPrszero(Q_1, Q_2), \\[2pt] \{i,j\} \in NE^{(t)} } \right\} \right|
\;\; \ge \;\; 
\frac{1}{4} \cdot \left| \NFPrszero(Q_1, Q_2) \right|.
$$
The statement of this lemma follows from the above inequality and the fact that 
$$\left| \left\{ \substack{ \{i,j\} \in \NFPrszero(Q_1, Q_2), \\[2pt] \{i,j\} \in \OP{NFbd} } \right\} \right| \;\; \ge \;\; \left| \left\{ \substack{ \{i,j\} \in \NFPrszero(Q_1, Q_2), \\[2pt] \{i,j\} \in NE^{(t)} } \right\} \right|.$$
\end{proof}
\end{appendix}